\documentclass[MNA]{mna} %

\SHORTTITLE{Gabor frames from contact geometry in models of the primary visual cortex.}

\TITLE{Gabor frames from contact geometry in models of the primary visual cortex.
   } 

\AUTHORS{%
  Vasiliki~Liontou\footnote{Department of Mathematics, University of Toronto,  ON  M5S 2E4, Canada.
    \BEMAIL{vasiliki.liontou@mail.utoronto.ca}}
  \and 
 Matilde~Marcolli\footnote{Department of Mathematics and Department of Computing and Mathematical Sciences, 
California Institute of Technology, Pasadena, CA 91125, USA.
 	\BEMAIL{matilde@caltech.edu}}}

\SHORTAUTHOR{V.~Liontou and M.~Marcolli}

\KEYWORDS{Visual cortex;  contact geometry; Gabor frames} %

\AMSSUBJ{53Z10} 
\AMSSUBJSECONDARY{92C20; 42C15; 57K33; 53D10} 

\SUBMITTED{July 6, 2022} 
\ACCEPTED{April 12, 2023} 

\ARXIVID{2111.02307} 

\VOLUME{3}
\YEAR{2023}
\PAPERNUM{2}
\DOI{10.46298/mna.9766}

\ABSTRACT{We introduce a model of the primary visual cortex $(V_1)$, which allows the compression and decomposition of a signal by a discrete family of orientation and position dependent receptive profiles. We show in particular that a specific framed sampling set and an associated Gabor system is determined by the Legendrian circle bundle structure of the $3$-manifold of contact elements on a surface (which models the $V_1$-cortex), together with the presence of an almost complex structure on the tangent bundle of the surface (which models the retinal surface). We identify a maximal area of the signal planes, determined by the retinal surface, that provides a finite number of receptive profiles, sufficient for good encoding and decoding. We then consider a $5$-dimensional model where receptive profiles also involve a dependence on frequency and scale variables, in addition to the dependence of position and orientation. In this case we show that the proposed window function does not give rise to frames (even in a distributional sense), while a natural modification of the same window generates Gabor frames with respect to the appropriate lattice determined by the contact geometry.}

\newcommand{\bS}{\mathbb{S}}
\newcommand{\C}{\mathbb{C}}
\newcommand{\R}{\mathbb{R}}
\newcommand{\Z}{\mathbb{Z}}
\newcommand{\cB}{\mathcal{B}}
\newcommand{\cC}{\mathcal{C}}
\newcommand{\cE}{\mathcal{E}}
\newcommand{\cF}{\mathcal{F}}
\newcommand{\cG}{\mathcal{G}}
\newcommand{\cI}{\mathcal{I}}
\newcommand{\cQ}{\mathcal{Q}}
\newcommand{\cS}{\mathcal{S}}
\newcommand{\GL}{\rm{GL}}
\newcommand{\Hom}{\rm{Hom}}
\newcommand{\Spec}{\rm{Spec}}

\begin{document}
	


\section{Introduction}\label{intro}

Two interesting mathematical models account for the functional architecture of the $V_1$ visual cortex,
both of them originally developed in the '80s and considerably expanded and refined in recent years: 
a model of receptive profiles in terms of Gabor functions,~\cite{Daug,Daug2,Marce},
and a model of the connectivity and the hypercolumn structure of the $V_1$ cortex in terms of
contact geometry and contact bundles~\cite{Hoff}. These two aspects of the mathematical 
modeling
of the visual cortex may appear at first unrelated, the first capturing functional analytic aspects of signal
encoding in terms of the neurons receptive profiles, the latter describing the geometric structure of
the visual cortex that captures the sensitivity to orientation of the simple cells in the hypercolumns. 
The fiber bundle contact geometry also provides a good geometric description of the connections 
between simple cells in different hypercolumns. These two mathematical models are in fact closely
entangled, as the more recent works of Petitot and Tondut~\cite{PeTon}, Citti and 
Sarti~\cite{CiSa}, and 
Sarti et al.~\cite{SaCiPe} have clearly shown. The simple cells profile shapes and their
geometric arrangement in the hypercolumn structure are simultaneously governed by the
same action of the rototranslation group, combined with a principle of selectivity of maximal response
(see~\cite{CiSa}). Thus, it appears that the contact geometry of the $V_1$ cortex also
determines its signal analysis properties. This is an interesting mathematical observation in itself,
that certain classes of contact manifolds carry an associated signal analysis framework entirely 
determined by the geometry. Part of the purpose of the present paper is to clarify what this
means in the specific case of contact $3$-manifolds that are Legendrian circle bundles over a
surface, and contact $5$-manifolds obtained from them by symplectization and contactization,
which are the two cases of direct relevance to the neuroscience modeling. Our main
focus here is on identifying additional aspects of the contact geometry that have a direct influence
on the signal analysis properties, beyond the relations already identified in previous work
such as~\cite{SaCiPe}. In particular, while previous work focused on continuous 
representations
of signals through short time Fourier transform, we argue that a more refined model should
incorporate the discrete nature of the neuron population involved, and identify a mechanism 
that ensures a good signal encoding and decoding in terms of a selection of a 
discrete system of filters. We argue that this selection of a discrete Gabor system with
adequate signal analysis properties can also be seen as directly encoded in the 
geometric model of the $V_1$-cortex. Our key observation to this purpose is the
fact that the combined presence of the contact structure on the Legendrian circle bundle
and a complex structure on the base surface determines an associated
bundle of framed lattices, which in turn provide the required discrete sampling set for
the Gabor frames.

Gabor filters play an essential role to both neural modeling and signal processing. In the works 
of Daugman~\cite{Daug, Daug2} and Marcelja~\cite{Marce}, it is argued why Gabor filters are 
the right 
choice 
for the modeling of receptive profiles of visual neurons in $V_1$. In particular, simple cells of 
the primary visual cortex try to localize at the same time the position $(x,y)$ and the 
frequency $w$ of a signal detected in the retina. However, the uncertainty principle in signal 
analysis indicates that it is impossible to detect both position and frequency with arbitrary 
precision.  Gabor filters minimize the uncertainty and therefore they process spatiotemporal 
information optimally.  Thus, a receptive profile, centered at $(x_0,y_0)$, with preferred 
spatial frequency $w=\sqrt{ u^2_0+v^2_0 }$ and preferred orientation 
$\theta=\arctan(\frac{v_0}{u_0})$  is efficiently modelled 
by a bivariate, real-valued Gabor  function  $f(x,y)$ of the form 
$\exp(-\pi((x-x_0)^2+(y-y_0)^2))\exp(-2\pi i(u_0(x-x_0)+u_0(y-y_0)))$. Given a distribution 
$I(x,y)$, specifying the distribution of light intensity of a visual stimulus, the receptive profile 
generates the response to that distributed stimulus via integration
\[
\mathrm{response} = \int\int_{-\infty}^{+\infty} I(x,y)f(x,y)dxdy.
\]

The integral representing the response of a receptive field is commonly used in time-frequency analysis, as short time Fourier transform. In a Euclidean space $\R^d$ of arbitrary dimension, 
the short time Fourier transform of a signal $I$ with respect to a window function $g$ is a linear and continuous, joint time-frequency representation defined as
\[V_gI(x,w)=\int_{\mathbb{R}^d}I(t)\overline{g(t-x)}e^{-2\pi i t\cdot w}  dt,\quad\mbox{ for 
}\quad x,w\in\mathbb{R}^d.
\]
More specifically, in the plane $\R^2$, the response of a receptive profile to a visual signal $I$ 
is equal to the short time Fourier transform of the signal $I$ with respect to the Gaussian 
$g(x,y)=\exp(-\pi(x^2+y^2))$, multiplied with a complex exponential
\[
\mathrm{response} = e^{2\pi i (x_0,y_0)\cdot(u_0,v_0)}V_gI(x_0,y_0,u_0,v_0).
\]

The short time Fourier transform is suitable for most theoretical approaches of
 space-frequency and time-frequency analysis. However, it is not practical to use continuous 
 representations for experimental purposes, when dealing with a finite (albeit large) population 
 of neurons. Continuous representations of signals, like the short time Fourier transform, allow 
 good encoding and decoding of the signal by using an uncountable system of receptive 
 profiles. Methods from discrete time-frequency analysis come to solve this problem. In 
 discrete methods, a discrete system of Gabor elementary functions is enough to reconstruct 
 and deconstruct the signal. If the window function is supported on a subset of the ambient 
 Euclidean space centered in $(x,y)$, the STFT $V_gI(x_0,y_0,u_0,v_0)$ carries the same 
 information for neighbouring points in the support of $g$ and therefore it is possible to 
 reduce the sampling set without compromising the quality of encoding and decoding of the 
 signal. While there is rich bibliography on the representation of receptive profiles by 
 continuous time-frequency signal representations, the question that arises is whether the 
 functional geometry of the visual mechanism directly incorporates a choice of a discrete 
 sampling set suitable for decoding and encoding visual stimuli.

An approach to modelling geometrically the functional architecture of the $V_1$ visual cortex
in terms of contact and sub-Riemannian geometry was developed by
 Petitot, and by  Citti and  Sarti,~\cite{Pet, SaCiPe, SaCi-book}.
The purpose of our note here is to highlight some aspects of the contact geometry of the
visual cortex, with special attention to a geometric mechanism for the generation of
families of Gabor frames. These give rise to a signal analysis setting that is adapted to
the underlying contact geometry. We focus here on the specific $3$-dimensional case of 
the manifold of contact elements of a $2$-dimensional surface, as this is the setting
underlying the model of~\cite{SaCiPe}. We also
discuss the case of an associated $5$-dimensional contact manifold considered i
n~\cite{BaspSartiCitti}. We will not discuss in this paper the more general question of
Gabor frames on arbitrary contact manifolds, which we plan to develop elsewhere, since
our main goal here is only to investigate some specific geometric aspects of the visual
cortex model developed in~\cite{SaCiPe} and in~\cite{BaspSartiCitti}.

Replacing the flat planar $\R^2$ as the domain of visual signals with a more general
curved Riemann surface $S$ is motivated by the fact that the retina is fixed to the eyeball,
hence not flat, that the resolution is not constant (thinner at the center than at the
periphery), and that the retinotopic map from the retina to the $V_1$ cortex along
the retino-geniculate-cortical pathway is a conformal map. Thus, the conformal 
geometry of a Riemann surface is a more suitable model than the flat linear
geometry of $\R^2$.

The contact $3$-manifold underlying the model of the $V_1$ cortex of~\cite{Pet, SaCiPe}
is of the form $M=\bS(T^* S)$, namely the unit sphere bundle of the cotangent bundle
of a $2$-dimensional surface $S$, also known as the manifold of contact elements 
of $S$. One of our main observations here is that the Legendrian circle bundle structure 
of $M$, together with the existence of an almost complex structure on the tangent 
bundle $TS$, provide a natural choice of a framed lattice (a lattice together with
the choice of a basis) on the bundle $\cE\oplus \cE^\vee$ over the
contact $3$-manifold manifold $M$, where $\cE$ is the pull-back of $TS$ to $M$.
This lattice determines an associated Gabor system, which has the general
form of the Gabor filters considered in~\cite{SaCiPe}. Using the complex
analytic method of Bargmann transforms, we investigate when the frame
condition is satisfied, so that one obtains Gabor frames for signal analysis
consistently associated to the fibers of $\cE$. 
\par In terms of the geometric model of the $V_1$ visual cortex, this shows that
the contact geometry directly determines the signal analysis, the
Gabor frames property, and the observed shape of the receptive profiles of the $V_1$ neurons. 
We show, in particular, that the window function proposed in~\cite{SaCiPe} to model the
receptive profiles, together with a scaling of the framed lattice determined by the
injectivity radius function of the surface $S$ (representing the retinal surface), give rise
to a Gabor system on the bundle of signal planes on the contact $3$-manifold $M$ (which
models the $V_1$ cortex) that satisfies the frame condition, hence has optimal signal
analysis properties.

On the other hand, the cortical simple cells are organized in hypercolumns, over each point 
$(x,y)$ of the retina, with respect to their sensitivity on a specific value of a visual feature. 
These features include orientation, color, spatial frequency, etc. In this context, the 
hypercolumnar architecture of $V_1$, for more than one visual feature, is modeled by a fiber 
bundle of dimension higher than $3$ over the retina. Each visual feature considered adds one 
more dimension to the fibers of the bundle. Thus, for the process of signals from an extended 
model, which includes more features than the three-dimensional orientation-selectivity 
framework, it is essential that higher dimensional models have optimal signal analysis 
properties. In~\cite{BaspSartiCitti}, Baspinar et al. extend the orientation selective model to 
include spacial frequency and phase. However, we show that the lift of the window
function, proposed in~\cite{SaCiPe} for the 3-dim model, 
to the $5$-dimensional contact manifold given by the contactization of the
symplectization of $M$, in the form proposed in~\cite{BaspSartiCitti}, only defines a
Gabor system in a distributional sense, and cannot satisfy the frame condition even
distributionally. We show that a simple modification of the proposed window function 
of~\cite{BaspSartiCitti} restores the desired Gabor frame property and allows for good
signal analysis in this higher dimensional model.

\section{Signals on manifolds of contact elements}

In this section we present the main geometric setting, namely a 
contact manifold that is either a $3$-manifold $M$
given by the manifold of contact elements of a compact $2$-dimensional surface, or the
$5$-manifold given by the contactization of the symplectization of $M$.  These
are, respectively, the geometries underlying the models of~\cite{PeTon}, and 
of~\cite{Pet, SaCiPe},
and the model of~\cite{BaspSartiCitti}.

The main aspect of the geometry that will play a crucial role in our construction of
the associated Gabor frames is the fact that these contact $3$-manifolds are endowed
with a pair of contact forms $\alpha, \alpha_J$ related through the almost-complex
structure $J$ of the tangent bundle $TS$. They have the property that the circle fibers
are Legendrian for both contact forms, while the Reeb vector field of each is
Legendrian for the other. This leads to a natural framing, namely a natural choice of
a basis for the tangent bundle $TM$, completely determined by the contact geometry.
It consists of the fiber direction $\partial_\theta$ and the two Reeb vector fields 
$R_\alpha$, $R_{\alpha_J}$.

\subsection{Legendrian circle bundles}

The results we discuss in this section apply, slightly more generally, to the case
of a $3$-manifold $M$ that is a Legendrian circle bundle over a 
2 dimensional compact surface $S$.

The Legendrian condition means that the fiber directions $T S^1$ inside the tangent
bundle $T M$ are contained in the contact planes distribution $\xi \subset TM$. Such
Legendrian circle bundles over surfaces are classified, see~\cite[p.~179]{Lutz}.
They are all either given by the unit cosphere bundle $M=\bS(T^* S)$, with the
contact structure induced by the natural symplectic structure on the cotangent
bundle $T^* S$, or by pull-backs of the contact structure on $M$ to a 
$d$-fold cyclic covering $M'\to M$, that exists for $d$ dividing $2g-2$, 
where $g=g(S)$ is the genus of $S$.
The case of $M=\bS(T^* S)$ is the manifold of contact elements of $S$. In the
following, we will restrict our discussion to this specific case.

In the geometric models of the $V_1$ cortex developed in~\cite{Pet, SaCiPe, SaCi-book},
the surface $S$ represents the retinal surface, while the fiber direction in the Legendrian circle 
bundle $M=\bS(T^* S)$ represents an additional orientation variable, which keeps track of
how the tangent orientation in $TS$ of a curve in $S$ is lifted to a propagation curve in the 
visual cortex,
where a line is represented by the envelope of its tangents rather than as a set of points.

The fibers of the sphere bundle $\bS(T^* S)$ are unit circles $S^1$, hence they can be
seen as parameterizing directions, that is, (oriented) lines in the plane $\R^2\simeq T^*_{(x,y)}S$.
One can also identify the circles with copies of $\mathbb{P}^1(\R)$ parameterizing lines in the
plane. This would correspond to considering the projectivized cotangent bundle instead
of the unit sphere bundle. While these two models are topologically equivalent in 
dimension $n=2$, they differ when considering the sub-Riemannian geometry 
of the rototranslation group $SE(2)$ as model geometry for the neural connectivity of 
the $V_1$ cortex, as in~\cite{CiSa, SaCiPe}.

\subsection{Liouville tautological \texorpdfstring{$1$}{1}-form and almost-complex 
twist}\label{Jsec}

Given a manifold $Y$, the cotangent bundle $T^*Y$ has a canonical
Liouville $1$-form, given in coordinates by $\lambda =\sum_i p_i\, dx^i$,
or intrinsically as $\lambda_{(x,p)}(v)=p(d\pi(v))$ for $v\in T_{x} Y$ and
$\pi: T^* Y \to Y$ the bundle projection. The canonical symplectic form
on $T^*Y$ is $\omega=d\lambda$.

Given an almost complex structure $J$ on $Y$, namely a $(1,1)$
tensor $J$ with $J^2=-1$, written in coordinates as 
$J=\sum_{k,\ell} J^k_\ell dx^\ell\otimes \partial_{x_k}$, the twist
by $J$ of the tautological Liouville $1$-form on $T^* Y$ is given by
\[ 
\lambda_J := \sum_{k,\ell} p_k J^k_\ell dx^\ell, 
\]
the $2$-form $\omega_J =d \lambda_J$ satisfies
\[
 \omega_J (\cdot,\cdot)=\omega( \hat J\cdot, \cdot), 
 \]
where in local coordinates
\[ \hat J=\begin{pmatrix} J^i_j & 0 \\ \sum_k p_k (\partial_{x_j} J^k_i -\partial_{x_i} J^k_j) & 
J^j_i 
\end{pmatrix}, \]
see for instance~\cite{Bertrand}.

In particular, in the case of a Riemann surface $S$, with coordinates $z=x+iy$ on $S$
and $p=(u,v)$ in the cotangent fiber, the tautological $1$-form is locally of the form
$\lambda =u\, dx + v\, dy$, with $\omega=du\wedge dx+ dv \wedge dy$.
The twisted tautological form with respect to $J$
given by multiplication by the imaginary unit, $J: (u,v)\mapsto (-v, u)$, given
by $\lambda_J= -v\, dx + u\, dy$, with $\omega_J = -dv\wedge dx +du \wedge dy$.

\begin{proposition}\label{JtwistLem}
On the contact $3$-manifold $M_w=\bS_w(T^*S)$, given by the cosphere bundle of radius $w$,
consider the contact $1$-form $\alpha$ induced by the tautological Liouville $1$-form $\lambda$
and the contact $1$-form $\alpha_J$ determined by the twisted $\lambda_J$. The
contact planes of these two contact structures intersect along the circle direction $\partial_\theta$.
The Reeb field $R_\alpha$ of $\alpha$ is Legendrian for $\alpha_J$ and the Reeb field $R_{\alpha_J}$
is Legendrian for $\alpha$. 
The twist $J$ fixes the $\partial_\theta$ generator and exchanges the generators $R_{\alpha_J}$
and $R_\alpha$.
\end{proposition}

\begin{proof}
On the contact $3$-manifold $M_w=\bS_w(T^*S)$, given by the cosphere bundle of radius $w$, 
the contact $1$-form induced by
the tautological Liouville $1$-form $\lambda$, written in a chart $(U,z)$ on $S$ with
local coordinate $z=x+iy$, is given by 
\begin{equation}\label{alphaM}
 \alpha =w \cos(\theta) dx + w \sin(\theta) dy, 
\end{equation} 
where $(w,\theta)$ are the polar coordinates in the cotangent fibers, and the corresponding
contact planes distribution on $U\times S^1_w$ is generated by the vector fields $\partial_\theta$ and
$-\sin(\theta) \partial_{x} + \cos(\theta) \partial_y$, and with the Reeb vector field
\begin{equation}\label{Ralpha}
R_\alpha =w^{-1} \cos(\theta) \partial_{x} + w^{-1} \sin(\theta) \partial_y
\end{equation}

The contact structure on $M_w$ induced by the twisted Liouville $1$-form $\lambda_J$ is given in
the same chart $(U,z)$  by
\begin{equation}\label{alphaJM}
 \alpha_J = - w \sin(\theta) dx + w \cos(\theta) dy 
\end{equation} 
with contact planes spanned by $\partial_\theta$ and $\cos(\theta) \partial_{x} + \sin(\theta) 
\partial_y$
and with Reeb vector field 
\begin{equation}\label{RalphaJ}
R_{\alpha_J}=- w^{-1} \sin(\theta) \partial_{x} + w^{-1} \cos(\theta) \partial_y.
\end{equation}
\end{proof}

\subsection{Symplectization and contactization}\label{SymplContSec}

Given a contact manifold $(M,\alpha)$, with $\alpha$ a given contact $1$-form,
one can always form a symplectic manifold $(M\times \R,\omega)$ with
$\omega=d(e^s \cdot \alpha)$ with $s\in \R$ the cylinder coordinate. Setting
$w=e^s\in \R^*_+$, one has $\omega= dw \wedge \alpha + w\, d\alpha$ on
$M\times \R^*_+$. In particular, in the case of the contact manifold $M=\bS(T^* S)$
this gives the following.

\begin{lemma}\label{symplem}
The complement of the zero section $T^*S_0:= T^*S\smallsetminus \{0\}$ 
is the symplectization of the manifold of contact elements $\bS(T^* S)$, with symplectic form
written in a chart $(U,z)$ of $S$ with $z=x+iy$ in the form
\begin{equation}\label{omegaTstarS}
 \omega = dw \wedge \alpha + w\, d\alpha = du\wedge dx+ dv \wedge dy 
\end{equation} 
or with the twisted contact and symplectic forms, given in the same local chart by
\begin{equation}\label{omegaJTstarS}
 \omega_J =dw \wedge \alpha_J + w\, d\alpha_J = -dv\wedge dx +du \wedge dy. 
\end{equation} 
for $(u,v)=(w\cos\theta, w\sin\theta)$.
\end{lemma}

Given a symplectic manifold $(Y,\omega)$, if the symplectic form is exact, $\omega=d\lambda$,
then one can construct a contactization $(Y\times S^1, \alpha)$ with $\alpha = \lambda-d\phi$,
where $\phi$ is the angle coordinate on $S^1$. When the symplectic form is not exact, it is
possible to construct a contactization if there is some $\hbar >0$ such that the differential form 
$\omega/\hbar$ defines an integral cohomology class, $[\omega/\hbar]\in H^2(Y,\Z)$. In this case
there is a principal $U(1)$-bundle $\cS$ on $Y$ with Euler class $e(\cS)=[\omega/\hbar]$,
endowed with a connection $\nabla$ with curvature $\nabla^2=\omega/\hbar$. This is also
known as the \emph{prequantization} bundle. This
connection determines a $U(1)$-invariant $1$-form $\alpha$ on $\cS$. The non-degeneracy
condition for the symplectic form $\omega$ implies the contact condition for the $1$-form $\alpha$. 
Different choices of the potential $\alpha$ of the connection $\nabla$ lead to equivalent contact manifolds
up to contactomorphisms, see~\cite{ElHoSa} for a brief summary of symplectization and 
contactization.

\begin{lemma}\label{contsymplem}
The contactization of the symplectization of the contact $3$-manifold $M=\bS(T^* S)$
is the $5$-manifold $T^*S_0 \times S^1$ with the contact form 
\[
 \tilde\alpha =\lambda- d\phi= w\alpha -d\phi . 
\]
\end{lemma}

\begin{proof}
The symplectization of a contact manifold is an exact symplectic manifold, hence it admits
a contactization in the simpler form described above. Thus, starting with the contact
manifold $M=\bS(T^* S)$ for a $2$-dimensional compact surface $S$, endowed with the contact
form $\alpha$ as in~\eqref{alphaM} that makes $M$ a Legendrian circle bundle, one obtains
the symplectization $T^*S_0$ with Liouville form $\lambda= w\alpha$, $w=e^s\in \R^*_+$,
and the contactization of the resulting exact symplectic manifold $(T^*S_0, \omega=d\lambda)$
is given by $T^*S_0 \times S^1$ with the contact form $\tilde\alpha =\lambda- d\phi= w\alpha -d\phi$.
\end{proof}

\begin{remark}\label{twistrem} 
The twist $\alpha \mapsto \alpha_J$ of~\eqref{alphaJM} of the contact structure on 
$M=\bS(T^*S)$
induces corresponding twists of the symplectization $\omega\mapsto \omega_J$ as 
in~\eqref{omegaJTstarS} and $\tilde\alpha \mapsto \tilde\alpha_J = w\alpha_J -d\phi$.  
\end{remark}

\begin{definition}\label{contsympldef} 
We write 
\begin{equation}\label{symplcont}
\cS(M):=T^*S_0  \ \ \ \text{ and } \ \ \ \cC\cS(M):=T^*S_0\times S^1,
\end{equation}
for the symplectization $\cS(M)$ of $M=\bS(T^*S)$ and the contactization $\cC\cS(M)$ 
for this symplectization, endowed with the contact and symplectic forms described above. 
\end{definition}

In the context of geometric models of the $V_1$ cortex, the $5$-dimensional contact
manifold $\cC\cS(M)$ corresponds to the model for the receptive fields considered 
in~\cite{BaspSartiCitti}, where an additional pair of dual variables
is introduced, describing phase and velocity of spatial wave propagation.

\subsection{The bundle of signal planes}\label{Evec}

In the model of receptive profiles in the visual cortex (see~\cite{Pet, SaCiPe}),
signals are regarded as functions on the retinal surface and the receptive profiles
are modelled by Gabor filters in these and dual variables. When taking into
account the underlying geometric model, however, one needs to distinguish between
the local variables $(x,y)$ on a chart $(U,z=x+iy)$ on the surface $S$ (or the local variables $(x,y,\theta)$ on
the $3$-manifold $M$) and the linear variables in its tangent space $T_{(x,y)} S$. 
Thus, we think of the retinal signal
as a collection of compatible signals in the planes $T_{(x,y)} S$, as $(x,y)$ varies in $S$.
We consider a real $2$-plane bundle on the $3$-manifold $M$ that describes this
geometric space where retinal signals are mapped.

\begin{definition}\label{signalplanedef} 
Let $\cE$ be the real $2$-plane bundle on
the contact $3$-manifold $M=\bS(T^*S)$ 
obtained by pulling back the  tangent bundle $TS$ of the surface $S$ to
$M$ along the projection $\pi: \bS(T^*S) \to S$ of the unit sphere bundle of $T^*S$,
\begin{equation}\label{Ebundle}
\cE=\pi^* TS .
\end{equation}
At each point $(x,y,\theta)\in M$, with $z=x+iy$ the coordinate in a local chart $(U,z)$ of $S$, 
the fiber $\cE_{(x,y,\theta)}$ is the same
as the fiber of the tangent bundle $T_{(x,y)} S$. 
Also let $\cE^\vee$ be the dual bundle of $\cE$, namely the bundle
of linear functional on $\cE$,
\[
 \cE^\vee =\Hom(\cE,\R). 
\]
\end{definition}

Locally the exponential map from $TS$ to $S$ allows for a comparison between
the description of signals in terms of the linear variables of $TS$ and the nonlinear variables of $S$.
The linear variables of $TS$ are the ones to which the Gabor filter analysis applies. Thus,
in terms of the contact $3$-manifold $M$, we think of a signal as a consistent family of
signals on the fibers $\cE_{(x,y,\theta)}$, or equivalently a signal on the total space of the $2$-plane
bundle $\cE$. The filters in turn will depend on the dual linear variables of $\cE$ and $\cE^\vee$.
We make this idea more precise in the next subsections. 

\subsection{Fourier transform relation and signals}

Over a compact Riemannian manifold $Y$, functions on the tangent and cotangent 
bundles $TY$ and $T^*Y$ are related by Fourier
transform in the following way. Let $\cS(TY,\R)$ denote the vector space of smooth real
valued functions on $TY$ that are rapidly decaying along the fiber directions, and similarly
for $\cS(T^*Y,\R)$. Let $\langle \eta,v \rangle_{x}$ denote the pairing of tangent and 
cotangent
vectors $v\in T_{x}Y$, $\eta\in T^*_{x}Y$ at a  point $x\in Y$. One defines
\begin{align*}
\cF: &\cS(TY,\R) \to \cS(T^*Y,\R) \\
	&{(\cF\varphi)}_{x}(\eta) = \frac{1}{{(2\pi)}^{\dim Y}} \int_{T_{x}Y} e^{2\pi i \langle \eta,v 
\rangle_{x}} 
\varphi_{x}(v)\, d{\rm vol}_{x}(v), 
\end{align*}
with respect to the volume form on $T_{x} Y$ induced by the Riemannian metric.

Because of this Fourier transform relation, cotangent vectors in $T^* Y$ are
sometimes referred to as ``spatial frequencies''.

In the model we are considering, the manifold over which signals are defined
is the total space $\cE$ of the bundle of signal planes introduced in Section~\ref{Evec} above,
namely real $2$-plane bundle $\cE=\pi^* TS$. We can easily generalize the
setting described above, by replacing the pair of tangent and cotangent bundle $TY$ and $T^* Y$ 
of a manifold $Y$ with a more general pair of a vector bundle $\cE$ and its dual $\cE^\vee$.
The variables in the fibers of $\cE^\vee$ are the spatial frequencies variables of the 
models of the visual cortex of~\cite{Pet, SaCiPe}. 
In this geometric setting a ``signal'' is described as follows.

\begin{definition}\label{signaldef} 
A signal  \(\cI\) is a real valued function on the total space $\cE$ of
the bundle of signal planes, with $\cI\in L^2(\cE,\R)$, with respect to the measure given by
the volume form of $M$ and the norm on the fibers of $\cE$ induced by the inner
product on $TS$ through the pull-back map. A smooth signal is a 
smooth function that decays to zero at infinity in the fiber directions, 
$\cI \in \cC^\infty_0(\cE,\R)$.  
\end{definition}

The assumption that $\cI$ is smooth is quite strong,
as one would like to include signals that have sharp contours and 
discontinuous jumps, but we can assume that such signals are smoothable
by convolution with a sufficiently small mollifier function that replaces
sharp contours with a steep but smoothly varying gradient.

\subsection{Signal analysis and filters}\label{signalfiltersSec}

For signals defined over $\R^n$, instead of over a more general manifold, 
signal analysis is performed through a family of filters (wavelets), and the
signal is encoded through the coefficients obtained by integration against the filters. 
Under good conditions on the family of filters, such as the frame condition for Gabor
analysis, both the encoding and the decoding maps are bounded operators, so the 
signal can be reliably recovered from its encoding through the filters.

For signals on manifolds there is in general no good construction of
associated filters for signal analysis, although partial results exist involving
splines discretization, diffusive wavelets, or special geometries such as
spheres and conformally flat manifolds, see for instance~\cite{BerKey, EbWi, Pese}.
One of our goals here is to show that geometric modelling of the visual cortex in terms of 
contact geometry and the description of receptive fields in terms of Gabor frames suggest a
general way of performing signal analysis on a specific class of contact manifolds.

The signal analysis model we propose in the following relies on encoding a signal
$f: S \to \R$ that is supported on a curved Riemann surface $S$ in terms of a 
function defined on the total space of the $2$-plane bundle $\cE$ over the $3$-dimensional
contact manifold $M=\bS(T^* S)$. The restrictions to the fibers of $\cE$ provide a
collection of signals defined on $2$-dimensional linear spaces, which describe
the lifts of the original signal $f: S \to \R$ to the local linearizations of $S$ given
by the fibers of the tangent bundle $TS$. The presence of the additional circle
coordinate $S^1$ in the $3$-manifold $M=\bS(T^* S)$ will account for the fact
that the Gabor filters used for signal analysis, which themselves live on the
liner fibers of $\cE$ include a directional preference specified by the angle
coordinate in the fibers of $\bS(T^* S)$. In terms of modelling of the visual
cortex, what we are presenting in Section~\ref{GaborContactEltsSec} below 
is a functional analytic model of the lifting of
signals from the (curved) retinal surface to linear spaces where the Gabor filters 
corresponding to the receptive fields of the $V_1$ neurons act to encode the
signal. In particular, as we discuss in Section~\ref{FrameSec} below, we will introduce a
version of geometric Bargmann transform. In our setting, since signals are
lifted from $S$ to the bundle $\cE$, the appropriate Bargmann transform is 
defined in terms of the duality of the bundles $\cE$ and $\cE^\vee$ over the 
contact $3$-manifold $M$. This version of geometric Bargmann transform 
differs from other versions previously
considered in~\cite{BarCi},~\cite{BarCiSS},~\cite{DuCi} constructed in terms
of the geometry of the Lie algebra of $SE(2)$, or in~\cite{Bar} where frame systems
are generated using unitary actions of discrete groups. It also differs from
other generalizations of the Bargmann transform such as in~\cite{Antoine}.

\section{Gabor filters on the manifold of contact elements}\label{GaborContactEltsSec}

In this section we present a construction of a family of Gabor systems associated to
the contact manifolds described in the previous sections.
As above we consider a compact Riemann surface $S$, and its
manifold of contact elements $M=\bS(T^* S)$ with the two 
contact $1$-forms $\alpha$ and $\alpha_J$ described in Section~\ref{Jsec} above.

\subsection{Gabor filters and receptive profiles}

As argued in~\cite{Daug2}, simple-cells in the $V_1$ cortex try to localize at the same time 
the position 
and the frequency of a signal, and the shape of simple cells is related to their functionality.
However, the uncertainty principle in space-frequency analysis implies that it is not possible to detect, 
with arbitrary precision, both position and momentum. At the same time, the need for the visual system 
to process efficiently spatio-temporal information requires optimal extraction and 
representation of images
and their structure. Gabor filters provide such optimality, since they minimize the uncertainty, and are
therefore regarded as the most suitable functions to model the shape of the receptive profiles.

The hypothesis that receptive field profiles are Gabor filters is
motivated by the analytic properties of Gabor frames. In addition to the minimization of the uncertainty
principle mentioned above, the frame condition for Gabor systems provides good encoding and
decoding properties in signal analysis, with greater stability to errors than in the case of a Fourier
basis. It is therefore a reasonable assumption that such systems would provide an optimal form
of signal analysis implementable in biological systems. We will be working here under the
hypothesis that receptive field profiles in the $V_1$ cortex are indeed Gabor filters. In this section
we show how to obtain such Gabor filters directly from the contact geometry described in the
previous section, while in the next section we discuss the frame condition.

\subsection{Gabor systems and Gabor frames}

We recall here the notion and basic properties of $d$-dimensional Gabor systems and Gabor frames,
see~\cite{Groch2}. Given a point $\lambda=(s,\xi)\in \R^{2d}$, with $s,\xi\in \R^d$, we 
consider the
operator $\rho(\lambda)$ on $L^2(\R^d)$ given by
\begin{equation}\label{rholambda}
 \rho(\lambda) := e^{2\pi i \langle s, \xi\rangle}\, T_s \, M_\xi 
\end{equation} 
with the translation and modulation operators
\begin{equation}\label{TMops}
 (T_s f)(t) = f(t-s), \ \ \ \ \ \ (M_\xi f)(t) = e^{2\pi i \langle \xi, t\rangle} f(t), 
\end{equation} 
which satisfy the commutation relation
\[
 T_s M_\xi = e^{-2\pi i \langle s,\xi\rangle} M_\xi T_s. 
\]

A Gabor system, for a given choice of a ``window function'' $g\in L^2(\R^d)$ and a 
$2d$-dimensional lattice $\Lambda=A \Z^{2d} \subset \R^{2d}$, for some $A\in \GL_{2d}(\R)$,
consists of the collection of functions
\begin{equation}\label{Gaborsys}
 \cG(g,\Lambda) ={\{ \rho(\lambda) g \}}_{\lambda\in \Lambda} .
\end{equation} 
More generally, Gabor systems can be defined in the same way for discrete sets $\Lambda \subset \R^{2d}$
that are not necessarily lattices. We will consider in this paper cases where the discrete set is a translate of
a lattice by some vector. In general, one assumes (see~\cite{Sei}) that the discrete set 
$\Lambda$ in the 
construction of the Gabor system is \emph{uniformly discrete}, namely such that
\[
 q(\Lambda) = \inf \{ \| \lambda - \lambda' \|\,|, \lambda, \lambda' \in \Lambda, \,\, \lambda \neq 
 \lambda' \} >0\, . 
\]
This is clearly satisfied in the case where $\Lambda$ is a translate of a lattice.

A Gabor system $\cG(g,\Lambda)$ as in~\eqref{Gaborsys} is a \emph{Gabor frame} if the
functions $\rho(\lambda) g$ satisfy the \emph{frame condition}: there are constants $C,C'>0$
such that, for all $h \in L^2(\R^d)$, 
\begin{equation}\label{framecond}
 C \, \| h \|^2_{L^2(\R^d)} \leq \sum_\lambda |\langle h, \rho(\lambda) g  \rangle |^2 \leq C' \, \| h \|^2_{L^2(\R^d)}. 
\end{equation}

The two inequalities in the frame condition ensure that both the encoding map that stores information
about signal $h$ into the coefficients $c_\lambda(h):=\langle h, \rho(\lambda) g  \rangle$ for $\lambda\in \Lambda$,
and the decoding map that reconstructs the signal from these coefficients are bounded linear operators.
This ensures good encoding and decoding, even though the Gabor frames $\{ \rho(\lambda) g \}$ 
do not form an orthonormal basis, unlike in Fourier analysis.

Window functions are typically assumed to have a Gaussian shape.
It is in general an interesting and highly nontrivial problem of signal analysis to characterize the
lattices $\Lambda$ for which the frame condition~\eqref{framecond} holds, for a given 
choice 
of window function,
see~\cite{Groch2}.

In the modelling of the $V_1$ cortex, receptive profiles are accurately modelled by Gabor functions,
hence it is natural to consider the question of whether there is a lattice $\Lambda$, directly determined
by the geometric model of $V_1$, with respect to which the receptive profiles are organized into a
Gabor frame system. This is the main question we will be focusing on in the rest of this paper.

\subsection{Window function}\label{WinSec}

The construction of Gabor filters we consider here follows closely the model of~\cite{SaCiPe},
reformulated in a way that more explicitly reflects the underlying contact geometry described
in the previous section. We first show how to obtain the mother function (window function) 
of the Gabor system and then we will construct the lattice that generates the  
system of Gabor filters.

Let $V$ and $\eta$ denote, respectively, the linear variables in the fibers $V\in T_{(x,y)}S\simeq \R^2$,
$\eta\in T^*_{(x,y)} S\simeq \R^2$, with $\langle \eta, V\rangle_{(x,y)}$ the duality pairing 
of $T^*_{(x,y)}S$ and $T_{(x,y)}S$. We write $V=(V_1,V_2)$ and $\eta=(\eta_1,\eta_2)$  
in the bases $\{ \partial_{x}, \partial_y \}$ and $\{ dx, dy \}$ of the
tangent and cotangent bundle determined by the choice of coordinates $(x,y)$ on $S$.

\begin{definition}\label{Psi0def}
A window function on the bundle $TS \oplus T^*S$ over $S$ is 
a smooth real-valued function $\Phi_0$ defined on the total space of
$TS \oplus T^*S$, of the form
\begin{equation}\label{Phi0S}
\Phi_{0,(x,y)}(V,\eta):= \exp\left(- V^t A_{(x,y)} V - i \langle \eta, V\rangle_{(x,y)} \right), 
\end{equation}
where $A$ is a smooth section of $T^*S \otimes T^*S$ that is symmetric 
and positive definite as a quadratic form on the fibers of $TS$, with the property that 
at all points $(x,y)$ in each local chart $U$ in $S$ the matrix $A_{(x,y)}$ has eigenvalues uniformly bounded away from zero,
$\Spec(A_{(x,y)})\subset [\lambda,\infty)$ for some $\lambda >0$. 
\end{definition}

\begin{lemma}\label{Psiwindowlem}
The restriction of a window function $\Phi_0$ as in~\eqref{Phi0S} to the bundle $TS \times 
\bS(T^*S)$ 
determines a real-valued function on the total space of the bundle $\cE$, which in a local chart is of the form
\begin{equation}\label{Psiwindow}
\Psi_{0,(x,y,\theta)}(V):= \exp\left(- V^t A_{(x,y)} V - i \langle \eta_\theta, V\rangle_{(x,y)} \right) .
\end{equation}
\end{lemma}

\begin{proof}
Consider the restriction of $\Phi_0$ to the bundle $TS \times \bS_w(T^*S)\subset TS\oplus T^*S$, 
for some $w>0$, over a local chart $(U,z=x+iy)$ of $S$. 
This means restricting the variable 
$\eta\in T^*_{(x,y)} S$ to $\eta=(\eta_1,\eta_2)=(w\cos(\theta),w\sin(\theta))$,
with $\theta\in S^1$,
\begin{equation}\label{Psi0res}
 \Phi_{0,(x,y)}|_{TS \times \bS_w(T^*S)}(V,\theta)=\exp\left(- V^t A_{(x,y)} V - i \langle \eta_\theta, V\rangle_{(x,y)} \right), 
\end{equation} 
with $\eta_\theta=(w\cos(\theta),w\sin(\theta))$. In particular, we restrict to the case $w=1$.

We can identify the total space of the bundle $TS \times \bS(T^*S)$ with the total space of
the bundle of signal planes $\cE$ over $M=\bS(T^*S)$. Indeed,
the direct sum of two vector bundles $E_1, E_2$ over the same base space $S$ is given by
\[
 E_1\oplus E_2=\{ (e_1,e_2)\in E_1 \times E_2\,|\, \pi_1(e_1)=\pi_2(e_2) \}. 
\]
Similarly, when considering sphere bundles
\[
 E_1\times \bS_w(E_2)=\{ (e_1,e_2)\in E_1 \times \bS_w(E_2) \,|\, \pi_1(e_1)=\pi_2(e_2) \}. 
\]
Consider the projection onto the second coordinate, $P: E_1\oplus E_2 \to E_2$. This
projection has fibers $P^{-1}(e_2)=\pi_1^{-1}(\pi_2(e_2))$. Thus, the total space of
the bundle $E_1\oplus E_2$, endowed with the projection $P$, can be identified with
the pull-back $\pi_2^* E_1$ over $E_2$, with fibers 
${(\pi_2^* E_1)}_{e_2}=\{ e_i\in E_1\,|\, \pi_1(e_1)=\pi_2(e_2) \}$, and similarly when
restricting to the sphere bundle of $E_2$.

Thus, we can write the function in~\eqref{Psi0res} equivalently as a real-valued function 
$\Psi_0$ on
the total space of the bundle $\cE$ over the contact $3$-manifold $M$, which is of the 
form~\eqref{Psiwindow}.
\end{proof}

This provides the reformulation of the Gabor profiles considered in~\cite{SaCiPe} in terms
of the underlying geometry of the bundle $\cE$ over $M$.

\subsection{Lattices} 

As above, consider the bundle of signal planes $\cE$ over $M=\bS(T^*S)$. The two
contact forms $\alpha$ and $\alpha_J$ discussed in Section~\ref{Jsec} determine
a choice of basis for $TM$ given by the Legendrian circle fiber direction $\partial_\theta$,
together with the two Reeb vector fields $R_\alpha$ and $R_{\alpha_J}$, each of which is
Legendrian for the other contact form. Over a local chart $U$ of $S$, these two
vector fields are given by~\eqref{Ralpha},~\eqref{RalphaJ} and
lie everywhere along the $TS$ direction,  
hence they determine a basis of the fibers $\cE_{(x,y,\theta)}$ of the bundle of signal planes
for $z=x+iy\in U$.

We denote by $\{ R_\alpha^\vee, R_{\alpha_J}^\vee \}$ the dual basis of
$\cE^\vee$ (over the same chart $U$ of $S$) characterized by $\langle R_\alpha^\vee , R_\alpha\rangle=1$, 
$\langle R_\alpha^\vee, R_{\alpha_J}\rangle=0$,
$\langle R_{\alpha_J}^\vee, R_\alpha\rangle =0$, $\langle R_{\alpha_J}^\vee,R_{\alpha_J}\rangle=1$. 
By the properties of Reeb and Legendrian vector fields, we can identify the dual basis with the contact forms, 
$\{ R_\alpha^\vee, R_{\alpha_J}^\vee \}=\{ \alpha, \alpha_J \}$.

Thus, the contact geometry of $M$ determines a canonical choice of a basis $\{ R_\alpha, R_{\alpha_J} \}$ 
for the bundle $\cE$ and its dual basis $\{ \alpha, \alpha_J \}$ for $\cE^\vee$.

This determines bundles of framed lattices (lattices with an assigned basis) over a local chart in $M$ of the form
\begin{equation}\label{latLambda}
\Lambda_{\alpha,J}:= \Z\, R_\alpha + \Z\, R_{\alpha_J}\, 
\end{equation}
\begin{equation}\label{latLambdavee}
\Lambda^\vee_{\alpha, J} := \Z\, \alpha + \Z\, \alpha_J \, . 
\end{equation}
where $\Lambda_{\alpha,J}$ and $\Lambda^\vee_{\alpha, J}$ here can
be regarded as a consistent choice of a lattice $\Lambda_{\alpha,J, (x,y,\theta)}$ 
(respectively, $\Lambda^\vee_{\alpha, J, (x,y,\theta)}$) in each fiber of $\cE$
(respectively, of $\cE^\vee$). The bundle of framed lattices
\begin{equation}\label{Lambdacomb}
 \Lambda_{\alpha,J} \oplus \Lambda^\vee_{\alpha, J} 
\end{equation} 
correspondingly consists of a lattice in each fiber of the bundle $\cE\oplus \cE^\vee$ over $M$.
We will also equivalently write the bundle of lattices~\eqref{Lambdacomb} in the form
$\Lambda + \Lambda_J$ with
\begin{equation}\label{LambdaJ}
\Lambda = \Z\, R_\alpha \oplus \Z\, \alpha, \ \ \  \Lambda_J =\Z\, R_{\alpha_J} \oplus \Z\, \alpha_J \, .
\end{equation}
In the following, we will often simply use the term ``lattice'' to indicate bundles of framed
lattices over $M$ as above.

\begin{lemma}\label{GaborEGamma}
The choice of the window function $\Psi_0$ described in Section~\ref{WinSec},
together with the lattice~\eqref{Lambdacomb}, determine a Gabor system
\[
 \cG(\Psi_0, \Lambda_{\alpha,J} \oplus \Lambda^\vee_{\alpha, J} ) 
\]
which consists, at each point $(x,y,\theta)\in M$ of the Gabor system
\[
 \cG(\Psi_{0,(x,y,\theta)}, \Lambda_{\alpha,J,(x,y,\theta)} \oplus \Lambda^\vee_{\alpha, 
 J,(x,y,\theta)}) 
\]
in the space $L^2(\cE_{(x,y,\theta)})$.
\end{lemma}

\begin{proof}
The Gabor functions in $\cG(\Psi_0, \Lambda+\Lambda_J)$ are of the form
\[
\rho(\lambda) \Psi_0=\rho(\xi) \rho(W)\, \Psi_0 = e^{2\pi i \langle \xi, V \rangle} \Psi_0(V-W), \]
for $\lambda=(\xi,W)$ with $\xi \in \Lambda^\vee_{\alpha,J}\subset \cE^\vee$ and $W \in \Lambda_{\alpha, J}\subset \cE$.
\end{proof}

\subsection{Injectivity radius function and lattice truncation}\label{scaleLSec}

In order to adapt this construction to a realistic model of signal processing in the $V_1$ cortex, one
needs to keep into account the fact that in reality only a finite, although large, number of Gabor filters in
the collection $\cG(\Psi_0, \Lambda+\Lambda_J)$ contribute to the analysis of the retinal signals. 
This number is empirically determined by the structure of the neurons in the $V_1$ cortex. This means
that there is some (large) cut-off size $R_{\max} >0$ such that the part of the lattice that 
contributes to the
available Gabor filters is contained in a ball of radius $R_{\max}$.

There is also an additional constraint that comes from the geometry. Namely, we are
using Gabor analysis in the signal planes determined by the vector bundle $\cE$ to analyze a
signal that is originally stored on the retinal surface $S$. Lifting the signal from $S$ to the fibers
of $\cE$ and consistency or results across nearby fibers is achieved through the exponential map
\[
 \exp_{(x,y)}: T_{(x,y)} S \to S 
\]
from the tangent bundle of $S$ (of which $\cE$ is the pull-back to $M$) to the surface. At a 
given
point $(x,y)\in S$ let $R_{inj}(x,y) >0$ be the supremum of all the radii $R>0$ such that the  
exponential map $\exp_{(x,y)}$ is a diffeomorphism on the ball $B(0,R)$ of radius $R$ in 
$T_{(x,y)}S$. For a compact surface $S$, we obtain a continuous 
\emph{injectivity radius function} given by $R_{inj}: S \to \R_+^*$ given by
$(x,y) \mapsto R_{inj}(x,y)$.

Thus, to obtain good signal representations and signal analysis in the signal planes,
we want that the finitely many available lattice points that perform the shift
operators $T_W=\rho(W)$ in the Gabor system construction lie within a ball of radius $R_{inj}$
in the fibers of $\cE$.

It is reasonable to assume that the maximal size $R_{\max}$, determined by empirical data on neurons in
the visual cortex, will be in general very large, and in particular larger than the maximum 
over the compact surface $S$ of the injectivity radius function. 
This means that, in order to match these two bounds, we need to consider a 
scaled copy of the lattice $\Lambda_{\alpha, J}$. 
We obtain the following scaling function.

\begin{lemma}\label{scaledlatt}
Let $b_M: M \to \R^*_+$ be the function given by 
\begin{equation}\label{afunctradii}
 b_M(x,y,\theta) :=\frac{R_{inj}(x,y)}{R_{\max}}\, ,
\end{equation}
where $R_{inj}(x,y)$ is the injectivity radius function and $R_{\max} >0$ is an assigned constant.
For $R_{\max} > \max_{(x,y)\in S} R_{inj}(x,y)$, consider the rescaled lattice 
\begin{equation}\label{aLambda}
\Lambda_{b,\alpha,J}:= b_M \, \Lambda_{\alpha, J} = \Z\, b_M\, R_\alpha + \Z\, b_M\, R_{\alpha_J} \ \ \ \text{ and } \ \ \
\left\{ \begin{array}{l} \Lambda_b = \Z \, b_M\, R_\alpha \oplus \Z\, \alpha, \\  \Lambda_{b,J} = 
\Z \, b_M\, R_{\alpha_J} \oplus \Z\, \alpha_J \end{array}\right..
\end{equation}
All the lattice points of the original lattice $\Lambda_{\alpha,J}$ that are within the
ball of radius $R_{\max}$ correspond to lattice points of the rescaled  $\Lambda_{b,\alpha,J}$
that are within the ball of radius $R_{inj}(x,y)$ in $\cE_{(x,y,\theta)}$. 
In particular, for $B$ a ball of measure $1$ in $\cE_{(x,y,\theta)}$, and $N(r)=\# \{ \lambda\in \Lambda_{b,\alpha,J}\cap r\cdot B \}$, we have
\begin{equation}\label{bDens}
D^-(\Lambda_{b,\alpha,J})= \liminf_{r\to \infty} \frac{N(r)}{r} = b_M^{-1} > 1 \, . 
\end{equation}
\end{lemma}

\begin{proof} The first statement is clear by construction. Moreover,
under the assumption that $R_{\max} > \max_{(x,y)\in S} R_{inj}(x,y)$, the
function $b_M$ of~\eqref{afunctradii} is everywhere smaller than one,
\begin{equation}\label{b1}
 b_M(x,y,\theta) <1, \ \ \ \forall (x,y,\theta)\in M \, , 
\end{equation} 
so that the density $D^-(\Lambda_{b,\alpha,J}) >1$.
\end{proof}

\begin{remark}\label{halfscale}
Note that we only need to rescale the $\Lambda_{\alpha, J}$ part of the lattice in $\cE$ and not 
the $\Lambda_{\alpha,J}^\vee$ part of the lattice in $\cE^\vee$, since the $\Lambda_{\alpha,J}^\vee$ part
only contributes modulation operators $M_\xi$ that do not move the coordinates outside of the
injectivity ball of the exponential map, unlike the translation operators $T_W$ with $W\in \Lambda_{\alpha, J}$.  
\end{remark}

We can also make the choice here to scale both parts of the lattice by the same factor $b=b_M$, and work
with the scaled lattice $\Lambda_{b,\alpha, J}\oplus \Lambda_{b,\alpha,J}^\vee$ even if the scaling of
the modulation part is not necessary by the observation of Remark~\ref{halfscale} above. The difference
between these two choices can be understood geometrically in the following way. One usually normalizes
the choice of the Reeb vector field of a contact form by the requirement that the pairing is 
$\langle \alpha, R_\alpha\rangle=1=\langle \alpha_J, R_{\alpha,J}\rangle$.
However, one can make a different choice of normalization. Scaling only the $\Lambda_{\alpha, J}$ part of the
lattice and not the $\Lambda_{\alpha,J}^\vee$ corresponds to changing this normalization, while scaling both
parts means that one maintains the normalization. As will be clear in the argument of Proposition~\ref{bFrameProp},
these two choices are in fact equivalent and give the same signal
analysis properties.

\section{The Gabor frame condition}\label{FrameSec}

In this section we check that the Gabor systems introduced above 
on the bundle of signal spaces $\cE$ satisfy the frame condition.
This condition is necessary for discrete systems of Gabor filters
to perform good signal analysis, in the sense that signals can be
reconstructed from their measurements by the filters.
In the usual setting of Gabor systems with Gaussian window
on a single vector space $\R^n$, the frame condition has been 
extensively studied. However, while in the $1$-dimensional case 
the frame condition can be characterized in terms of a density
property for the lattice (\cite{Lyu, Sei}), in higher dimensions
the question of whether a Gabor frame with Gaussian window in 
$\R^n$ and a given lattice $\Lambda\subset \R^{2n}$ satisfies the
frame condition is generally open and very difficult to assess,
see~\cite{Groch2}. Since we are specifically interested here in the
$2$-dimensional case, we will follow the method developed
in~\cite{Groch2}, based on the Bargmann transform, adapted to
our geometric setting.

We discuss separately the case where, in a local chart $U$ in $S$, the quadratic form $A$ in the window function $\Psi_0$
is diagonal in the basis $\{ R_\alpha, R_{\alpha_J} \}$ and the general case where it is not diagonal. The first case has
the advantage that it reduces to one-dimensional Gabor systems, for which we can reduce the
discussion to a famous result of Lyubarski\v{\i} and Seip,~\cite{Lyu, Sei}, after the 
slightly different form of the window function is accounted for. The more general case can be
dealt with along the lines of the results of~\cite{Groch2} for $2$-dimensional Gabor systems. 
In particular, the analysis of the frame condition relies on the complex analytic technique of
Bargmann transform and sampling.

As discussed in Section~\ref{signalfiltersSec} above, the notion of geometric Bargmann 
transform
that we introduce here, for the purpose of investigating the frame condition, is defined
in terms of the geometry of the dual pair of vector bundles $\cE$ and $\cE^\vee$ over
the contact $3$-manifold $M=\bS(T^* S)$, since in our setting retinal signals $f: S \to \R$
are lifted to signals that live on the linear fibers $\cE$, with the angular coordinate of
the circle fibers of $\bS(T^* S)$ accounting for the directionality of the Gabor filters.

\subsection{Gabor frame condition}

Let $\cE$ be the bundle of signal planes on the contact $3$-manifold $M$ as above. Let
$\Psi_0$ be a window function, which we assume of the form~\eqref{Psiwindow}.
Suppose given a lattice bundle $\Lambda$, namely a bundle over $M$ with
fiber isomorphic to $\Z^4$, where the fiber $\Lambda_{(x,y,\theta)}$ is a lattice in
$(\cE\oplus \cE^\vee)_{(x,y,\theta)}$. We form the Gabor system $\cG(\Psi_0,\Lambda)$
as in Lemma~\ref{GaborEGamma}, with Gabor functions 
$\rho(\lambda_{(x,y,\theta)}) \Psi_0 |_{\cE_{(x,y,\theta)}}$, with 
$\lambda_{(x,y,\theta)}\in \Lambda_{(x,y,\theta)}$.

\begin{definition}\label{smoothGaborDef}
The Gabor system $\cG(\Psi_0,\Lambda)$ satisfies the smooth Gabor frame condition on $M$ 
if there
are smooth  $\R^*_+$-valued functions $C,C'$ on the local charts of $M$, such that the 
frame condition 
holds pointwise in $(x,y,\theta)$,
\begin{equation}\label{smoothGabor}
 C_{(x,y,\theta)}\, \| f \|_{L^2(\cE_{(x,y,\theta)})}^2 \leq
\sum_{\lambda_{(x,y,\theta)} \in \Lambda_{(x,y,\theta)}}
|\langle f, \rho(\lambda_{(x,y,\theta)}) \Psi_0 \rangle |^2 \leq C'_{(x,y,\theta)}\, \| f \|_{L^2(\cE_{(x,y,\theta)})}^2 \, .
\end{equation}
\end{definition}

Note that, although, the manifold $M$ is compact, so that globally defined continuous functions
$C,C': M\to \R_+$ would have a minimum and a maximum that are strictly positive and finite, in
the condition above we are only requiring that the functions $C,C'$ are defined on the local
charts, without necessarily extending globally to $M$. Indeed, since global vector fields on 
an orientable compact surface $S$ necessarily have singularities (unless $S=T^2$), the frame
condition will not in general extend globally, while it holds locally within each chart, with not
necessarily uniformly bounded $C,C'$. If these functions extend globally to $M$, then a
stronger global frame condition
\[ 
C_{\min}\, \| f \|_{L^2(\cE_{(x,y,\theta)})}^2 \leq
\sum_{\lambda_{(x,y,\theta)} \in \Lambda_{(x,y,\theta)}}
|\langle f, \rho(\lambda_{(x,y,\theta)}) \Psi_0 \rangle |^2 \leq C'_{\max}\, \| f 
\|_{L^2(\cE_{(x,y,\theta)})}^2  
\]
would also be satisfied, but one does not expect this to be the case, except in special cases
like the parallelizable $S=T^2$. In the case directly relevant to the modeling of the primary
visual cortex, one assumes that the retinal surface is represented by a chart $U\subset S$ with
$S=S^2$ a sphere.

\subsection{The diagonal case: dimensional reduction}

Consider first the case where the quadratic form $A$ in~\eqref{Psiwindow} is diagonal in the
basis $\{ R_\alpha, R_{\alpha_J} \}$ of the bundle $\cE$.

First observe that, in a local chart $U$ of $S$, the unit vector $\eta_\theta \in T^*S$ is in fact the vector 
$\eta_\theta=(\cos(\theta),\sin(\theta))$ in the basis $\{ dx, dy \}$, which is
the dual basis element $\alpha$, as in~\eqref{alphaM}. Thus, the
window function~\eqref{Psiwindow} used in~\cite{SaCiPe} is of the
form
\begin{equation}\label{windowform}
 \Psi_{0,(x,y,\theta)}(V) 
= \rho(\frac{1}{2\pi}(0,1)) \hat\Psi_{0,(x,y,\theta)}(V) 
\end{equation}
where
\begin{equation}\label{hatPsi}
 \hat\Psi_{0,(x,y,\theta)}(V):= \exp\left(- V^t A_{(x,y)} V \right) 
\end{equation} 
and $(0,-1)\in \Lambda$ is the covector $-\eta_\theta(x,y)= \alpha|_{(x,y,\theta)}$.
Thus, the Gabor system can be equivalently described as 
\[
 \cG( \Psi_0, \Lambda + \Lambda_J )= \cG(\hat \Psi_0, \hat\Lambda + \Lambda_J ) 
\]
\begin{equation}\label{hatLambda}
\hat\Lambda =\xi_0 + \Lambda =\{ (W,\xi)\in \cE\oplus\cE^\vee\,|\,  W\in \Z R_\alpha, \,\, \xi\in \xi_0 +\Z\alpha  \} \ \ \ \text{ with } \xi_0 :=- \frac{1}{2\pi}\alpha \in \cE^\vee \, . 
\end{equation}
Note that $\hat\Lambda$ is no longer a lattice (a discrete abelian \emph{subgroup} in each
fiber $\cE_{(x,y,\theta)} \oplus\cE^\vee_{(x,y,\theta)}$ in the local chart): it is however a uniformly discrete set given by the
translate $\xi_0 + \Lambda$.

\begin{lemma}\label{diagcaselem}
If the quadratic form $A$ in~\eqref{Psiwindow} is diagonal, $A={\rm diag}(\kappa_1^2, 
\kappa_2^2)$, in the
basis $\{ R_\alpha, R_{\alpha_J} \}$ of $\cE$ in a local chart, then the Gabor frame condition for 
$\cG(\Psi_0, \Lambda+\Lambda_J)$  is equivalent to
the frame condition for two uncoupled problems for the one-dimensional Gabor systems 
$\cG(\psi_0, \Lambda)$ and $\cG(\phi_0,\Lambda_J)$, with 
$\psi_0(V_1)=\exp( - \kappa_1^2 V^2_1 -i V_1 )$ and
$\phi_0(V_2)=\exp(-\kappa_2^2 V_2^2)$.
\end{lemma}

\begin{proof}
Given the duality pairing relations between the contact forms $\alpha$, $\alpha_J$ and their Reeb 
vector fields $R_\alpha$ and $R_{\alpha_J}$, if we write the vectors $V\in \cE_{(x,y,,\theta)}$ in
coordinates $V= V_1\, R_\alpha + V_2\, R_{\alpha_J}$ over the local chart, then 
the window function is written in the form
\[
 \Psi_{0, (x,y,\theta)} (V_1,V_2) = \exp( - \kappa_1^2 V^2_1 -i V_1 )\cdot \exp(-\kappa_2^2 
 V_2^2) =\psi_0(V_1) \cdot \phi_0(V_2), 
\]
and the Gabor system is of the form
\begin{align*}
(\rho(\lambda)\Psi_0) (V) =(\rho(\lambda_1) \psi_0)(V_1) \cdot (\rho(\lambda_2) \phi_0)(V_2)\\  
\lambda_1=(\xi_1,W_1)\in \Lambda \ \ \ \text{ and } \ \ \  \lambda_2=(\xi_2,W_2)\in \Lambda_J\, . 
\end{align*}
This means that, in this case, the Gabor frame condition problem for $\cG(\Psi_0, \Lambda+\Lambda_J)$ 
reduces to two uncoupled problems for the one-dimensional Gabor systems $\cG(\psi_0,\Lambda)$
and $\cG(\phi_0,\Lambda_J)$. The frame condition for $\cG(\Psi_0, \Lambda+\Lambda_J)$ is
satisfied iff it is satisfied for $\cG(\psi_0,\Lambda)$ and $\cG(\phi_0,\Lambda_J)$, where the first
problem, by the discussion above, is equivalent to the frame condition for the system 
$\cG(\hat\psi_0,\hat\Lambda)$ with 
$\hat\Lambda=\xi_0+\Lambda$ and $\hat\psi_0(V_1)=\exp( - \kappa_1^2 V^2_1)$.
\end{proof}

\begin{proposition}\label{noframes}
The functions in the Gabor system $\cG(\Psi_0, \Lambda+\Lambda_J)$ are \emph{not} 
frames. 
\end{proposition}

\begin{proof}
The second case above is a one-dimensional Gabor system with a Gaussian window function
$g(t)=e^{-\kappa^2 t^2}$ and the lattice $\Z^2$, while the first case is a one-dimensional Gabor system with a
modified window function of the form $g(t)=e^{-\kappa^2 t^2 - i a t}$ and the lattice $\Z^2$ or equivalently
a window function $\hat g(t)=e^{-\kappa^2 t^2}$ and the discrete set $(0,a)+\Z^2$.

For a lattice $\Lambda=A\Z^d$ with $A\in \GL_d(\R)$ the density is given by $s(\Lambda)=|\det(A)|$.
In particular it is $s(\Lambda)=1$ for the standard lattice $\Z^2$. The \emph{density theorem} 
for Gabor
frames,~\cite{Jan} (see also Proposition~2 of~\cite{Groch2}), states that if a Gabor system 
$\cG(g,\Lambda)$ 
is a frame in $L^2(\R^d)$ and the window is a rapid decay function $g\in \cS(\R^d)$, then necessarily 
$s(\Lambda)<1$. Thus, these one-dimensional Gabor systems are not frames, hence the original system
$\cG(\Psi_0, \Lambda+\Lambda_J)$ also does not satisfy the frame condition.
\end{proof}

On the other hand, the situation changes when one takes into account the scaling of the
lattice discussed in Section~\ref{scaleLSec}.

\begin{proposition}\label{yesframes}
Consider the rescaled lattices $\Lambda_{b,\alpha,J}$, $\Lambda_b$, 
$\Lambda_{b,J}$ of~\eqref{aLambda}. The system 
$\cG(\Psi_0, \Lambda_b+\Lambda_{b,J})$ does satisfy the frame condition.
\end{proposition}

\begin{proof}
The Gabor frame question for the system $\cG(\Psi_0, \Lambda_b+\Lambda_{b,J})$ reduces to the question
of whether the one-dimensional systems $\cG(\phi_0,\Lambda_{b,J})$ and 
$\cG(\hat\psi_0,\hat\Lambda_b)$ with $\hat\Lambda_b=\xi_0+\Lambda_b$ are frames.

In the case of one-dimensional systems, there is a complete characterization of when
the frame condition is satisfied,~\cite{Lyu, Sei, SeiWall}. This characterization
is obtained by reformulating the problem in terms of a complex analysis problem of
sampling and interpolation in Bargmann-Fock spaces. 
In the case of a Gaussian window function $\psi$ and a 
uniformly discrete set $\Lambda \subset \R^2$, it is proved in~\cite{Sei}
that the Gabor system $\cG(\psi,\Lambda)$ is a frame if and only if the lower Beurling
density satisfies $D^-(\Lambda)>1$, where 
\[
 D^-(\Lambda) = \lim_{r\to \infty} \inf \frac{N^-_\Lambda(r)}{r^2},  
\]
with $N^-_\Lambda(r)$ the smallest number of points of $\Lambda$ contained in a
scaled copy $r\cI$ of a given set $\cI\subset \R^2$ of measure one, with measure zero
boundary. The value $D^-(\Lambda)$ is independent of the choice of the set $\cI$.
In the case of a rank two lattice this corresponds to the condition $s(\Lambda)<1$, 
which is therefore also sufficient.

Thus, the one-dimensional systems
$\cG(\phi_0,\Lambda_{b,J})$ and $\cG(\hat\psi_0, \hat\Lambda_b)$ are frames
if and only if $s(\Lambda_{b,J})<1$ and $s(\Lambda_b)<1$, since the
translate $\hat\Lambda_b$ and $\Lambda_b$ have the same lower Beurling density.
Since the scaling function satisfies $b_M <1$ everywhere on $M$, as in~\eqref{b1},
we have seen in Lemma~\ref{scaledlatt} that 
these conditions are satisfied. It follows that the Gabor system
$\cG(\Psi_0, \Lambda_b+\Lambda_{b,J})$ is a frame. 
\end{proof}

\subsection{The non-diagonal case: Bargmann transform} 

In the more general case where the quadratic form in $\Psi_0$ is not
necessarily diagonal in the basis $\{ R_\alpha, R_{\alpha,J} \}$ in a local chart, the question of 
whether the Gabor system  $\cG(\Psi_0, \Lambda_b+\Lambda_{b,J})$ satisfies
the frame condition can still be reformulated in terms of sampling 
and interpolation in Bargmann-Fock spaces, see~\cite{Groch2}.

\subsubsection{Bargmann transform and Gabor frames}

The Bargmann transform of a function $f$ in $L^2(\R^n)$ is defined as
\begin{equation}\label{Bargmann}
    \cB f(z)=\int_{\R^n}f(t)e^{2\pi t\cdot z-\pi t^2-\frac{\pi}{2}z^2}\, dt \, , 
\end{equation}
where, for $z\in\mathbb{C}^{n}$ we write $z=x+iw$ for some $x,w\in \R^n$ and 
$z^2=(x+iw)\cdot(x+iw)=x\cdot x-w\cdot w+i2x\cdot w$ and $|z|^2=x^2+w^2$. 
It is a unitary transformation from $L^2(\R^n)$ to the Bargmann-Fock space $\cF^2_n$,
which consists of entire functions of $z\in \C^n$ with finite norm
\begin{equation}\label{normBarFock2}
\| F \|^2_{\cF^2_n}=  \int |F(z)|^2\,e^{-\pi |z|^2}\,dz \, < \infty \, ,
\end{equation}
induced by the inner product
\[
\langle F, G \rangle_{\cF^2_n} = \int_{\C^n} F(z)\, \overline{G(z)}\, e^{-\pi |z|^2}\, dz \, . 
\]
We also consider the Bargmann-Fock space $\cF^\infty_n$, which is the space of
entire functions on $\C^n$ with
\begin{equation}\label{normBarFock}
\| F \|^2_{\cF^\infty_n}=\sup_{z\in \C^n} | F(z) |\, e^{-\frac{\pi |z|^2}{2} } \, < \infty \, .
\end{equation}
There is a well known relation between the Bargmann transform and Gabor systems with
Gaussian window function, see for instance~\cite{Groch, Groch2}. In our setting, because of 
the
form~\eqref{windowform} of the window function, we need a simple variant of this relation
between Gabor systems and Bargmann transform which we now illustrate.

A set $\Lambda\subset \C^n$ is a \emph{sampling set} for $\cF_n^2$ if there are constants
$C,C'>0$, such that, for all $F\in \cF_n^2$, 
\[ 
C \cdot \| F \|^2_{\cF_n^2} \leq \sum_{\lambda\in \Lambda} | F(\lambda) |^2 e^{-\pi 
|\lambda|^2} \leq C' \cdot \| F \|^2_{\cF_n^2} \, . 
\]
A set $\Lambda\subset \C^n$ is a \emph{set of uniqueness} for $\cF^\infty_n$ if a function 
$F\in\cF^\infty_n$ satisfying $F(\lambda)=0$ for all $\lambda\in \Lambda$ must vanish
identically, $F\equiv 0$. 
For $\Lambda\subset \C^n$, let $\bar\Lambda=\{ \bar\lambda\,|\, \lambda\in \Lambda \}$.

We consider as in~\cite{Groch3} the modulation spaces $M^p(\R^n)$ as the space of
tempered distributions $f\in \cS^\prime(\R^n)$ with Gabor transform with bounded $L^p$ norm,
$\| V_\varphi f\|_p <\infty$, for all $\varphi\in \cS(\R^d)$, where
\[
 V_\varphi f=\langle f, M_\xi T_{x} \varphi\rangle =\int_{\R^d} f(t) \overline{\varphi(t-x)} 
e^{-2\pi i \xi\cdot t} \, dt \, .  
\]
Similarly, the modulation space $M^\infty(\R^n)$ is the space of
tempered distributions $f\in \cS^\prime(\R^n)$ with 
$\| V_\varphi f\|_\infty <\infty$, for all $\varphi\in \cS(\R^d)$.

\begin{proposition}\label{Bergequiv1}
Let $\Lambda\subset \C^n$ be a lattice and let $\phi(x)=e^{-\pi \, |x|^2}e^{-2\pi i \, a\cdot x}\in L^2(\R^n)$,
for some fixed $a\in \R^n$. 
Then the following conditions are equivalent.
\begin{enumerate}
    \item The Gabor system $\cG(\phi,\Lambda)$ is a frame.
    \item The set $\bar\Lambda_a:=\bar\Lambda+ ia$ is a sampling set for $\cF^2_n$.
    \item The set $\bar\Lambda_a$ is a set of uniqueness for $\cF^\infty_n$. 
\end{enumerate}
\end{proposition}

\begin{proof} 
For the proof of $1\iff 2$ it suffices to prove that 
\[
|\langle f,M_wT_{x}\phi \rangle|=|\cB(x-i(w+a))|e^{-\frac{\pi |(x-i(w+a))|^2}{2}}.
\]
We have
\begin{align*}
V_\phi f(x,w) & =\int_{\R^n}f(t)e^{-\pi (t-x)^2}e^{-2\pi i (a\cdot(t-x))}e^{-2\pi i (w\cdot t)}dt\\
& =e^{2\pi i (a\cdot x)}\int_{\R^n}f(t)e^{-\pi t^2+2\pi tx-\pi x^2} e^{-2\pi 
i(a+w)\cdot t} dt\\
& = e^{2\pi i a\cdot x}e^{-\pi i x\cdot 
(a+w)}e^{-\frac{\pi}{2}(x^2+(a+w)^2)}\int_{\R^n}\!\! f(t)e^{-\pi t^2}e^{2\pi 
t\cdot(x-i(w+a))} e^{-\frac{\pi}{2}(x-i(a+w))^2} dt\, .
\end{align*}
Moreover, for $z^{\prime}=x+i(w+a)$, 
\[V_{\phi}f(x,w)=e^{-\frac{\pi}{2}|z^{\prime}|^2}e^{-\pi i x\cdot \Im(z^{\prime})}e^{2\pi i (a\cdot 
x)}\cB f(\overline{z^{\prime}})\]
Thus, $|V_\phi f(x,w)|=|\cB\, f(\overline{z^{\prime}})|e^{-\frac{\pi}{2}|z^\prime|^2}=|\cB\, f(x-i(w+a))|e^{-\frac{\pi |(x-i(w+a))|^2}{2}}$.
Thus, we obtain 
\[
\sum_{\lambda\in \Lambda}|V_\phi f(\lambda)|=\sum_{z^{\prime}\in \Bar{\Lambda_a}}|\cB\, 
f({z^{\prime}})|e^{-\frac{\pi}{2}|z^\prime|^2}\, , 
\]
and $\sum_{\lambda\in \Lambda}|V_\phi f(\lambda)|\asymp \| f \|_{L^2(\R^n)}$ if and only if 
\[
\sum_{z^{\prime}\in \bar{\Lambda}_a}
|\cB\, f({z^{\prime}})| e^{-\frac{\pi}{2} |z^\prime|^2 } \asymp \| \cB\, f \|_{\cF^2_n}\, . 
\]
To prove $2\iff 3$, 
starting with the assumption that $\bar{\Lambda}_a$ is a set of sampling for $\cF_n^2$, 
let $F\in \cF^\infty_n$ be such that $F(\lambda)=0$ for all $\lambda\in \bar{\Lambda}_a$. 
The Bargmann-Fock space $\cF^\infty_n$
is related to the modulation space $M^\infty(\R^n)$ through the Bargmann 
transform~\eqref{Bargmann},
\[
 \cF^\infty_n =  \cB(M^\infty(\R^n))\, . 
 \]
Thus, there exists an element $f\in M^\infty(\R^n)$ such that $\cB \, f=F$. Thus, we have
$\cB\, f(\lambda) =0$, for all $\lambda\in \bar{\Lambda}_a$,  hence 
$\langle f,\pi(\lambda)\phi \rangle =0$, for all $\lambda\in \Lambda$. 
The equivalence $1\iff 2$ then implies that $f\equiv 0$, hence $F\equiv 0$.

Conversely, suppose that $\bar{\Lambda}_a$ is a set of uniqueness for $\cF^\infty_n$. 
Theorem~3.1 of~\cite{Groch3} shows that the frame
condition for the Gabor system $\cG(\phi,\Lambda)$, for a window $\phi\in \cS(\R^n)$, 
is equivalent to the condition that the Gabor transform map is one-to-one as a map 
\begin{equation}\label{Vgmap}
V_\phi : M^\infty(\R^n)\to \ell^\infty(\Lambda)\, , \ \ \ \ 
V_\phi: f \mapsto V_\phi f |_\Lambda  \, .
\end{equation}
Since we have $\phi\in \cS(\R^n)$, it suffices to prove that the Gabor transform
$f\mapsto  V_{\phi}f|_{\bar\Lambda_a}$ is one-to-one as a map 
$M^\infty(\R^n) \to \ell^\infty(\bar\Lambda_a)$.

Let $D$ denote the map $D: M^\infty(\R^n) \to \ell^\infty(\Lambda)$ given by
\[ D: f \mapsto \{ \cB\, f (\lambda) \}_{\lambda \in \Lambda}\, , \]
and let $T: \ell^\infty(\Lambda) \to \ell^\infty(\Lambda_a)$ be given by
\[ T : \{ c_\lambda \}_{\lambda \in \Lambda} \mapsto \{ 
e^{\pi i \lambda _1 (\lambda_2+a)}e^{-|\lambda+(0,a)|^2/2} c_\lambda \}_{ \lambda+(0,a) \in 
\Lambda_a}\, , \]
The operator $V_\phi$ of~\eqref{Vgmap} is the composite $V_\phi = T \circ D$, which is 
injective since
both $T$ and $D$ are. 
\end{proof}

\begin{remark}\label{sameframe}  In particular this shows that, with the
window functions $\tilde{\phi}(x)=e^{-\pi x^2}$ and 
$\phi(x)=e^{-\pi x^2}e^{-2\pi i (a\cdot x)}$, the Gabor system 
$\cG(\phi,\Lambda)$ is a frame if and only if $\cG(\tilde{\phi},\Lambda)$ is a frame. 
\end{remark}

Indeed, for the window $\tilde\phi$ the system $\cG(\tilde{\phi},\Lambda)$ is a frame
iff the system $\cG(\tilde{\phi},\Lambda_a)$ is a frame and the latter is equivalent to
\[
 \sum_{z\in \bar\Lambda} |\cB\, f({z-ia})|e^{-\frac{\pi}{2}|z|^2}\asymp ||\cB\, f||_{\cF_n^2}\, , 
\]
which we have seen is equivalent to $\cG(\phi,\Lambda)$ being a frame.

\subsubsection{Geometric Bargmann transform} 

We apply this Bargmann transform argument to our geometric setting. 
The bundle $\cE$ is endowed with an almost complex
structure, coming from the identification $\cE=\pi^* TS$ with $S$ a
Riemann surface, hence the dual $\cE^\vee$ can also be endowed
with an almost complex structure. However, for the purpose of applying 
the Bargmann transform argument in our setting, we just need to consider the
bundle $\cE \oplus \cE^\vee$ as a complex $2$-plane bundle over $M$.
First note that the local bases $\{ R_\alpha, R_{\alpha_J} \}$ of $\cE$ and
$\{ \alpha, \alpha_J \}$ of $\cE^\vee$ determine a local isomorphism
between $\cE$ and $\cE^\vee$.  For $(W,\eta)\in ({\cE \oplus \cE^\vee})_{(x,y,\theta)}$,
with $W=W_1 R_\alpha + W_2 R_{\alpha_J}$ and $\eta=\eta_1 \alpha + \eta_2 \alpha_J$,
we define $J: \cE \oplus \cE^\vee \to \cE \oplus \cE^\vee$ with $J^2=-1$ by setting
\[ J \, (W,\eta) := (\eta, -W) =\eta_1\, R_\alpha + \eta_2\, R_{\alpha_J} - W_1\, \alpha - W_2\, 
\alpha_J\, . \]
We can then take $W+ i \eta :=(W,\eta)$ with scalar multiplication by $\lambda\in \C$, $\lambda=x+iy$
with $x,y\in \R$ given by $\lambda \cdot (W+ i \eta)=(x+y\, J)\,  (W,\eta)$. This gives a
fiberwise identification
\begin{equation}\label{Imap}
 \cI: (\cE \oplus \cE^\vee)_{(x,y,\theta)} \stackrel{\simeq}{\to} \C^2\, \ \ \ 
 (W,\eta) \mapsto z=(z_1,z_2)=(W_1+i\eta_1, W_2+i\eta_2)\, . 
\end{equation}

Given the choice of a window function $\Psi_{0,(x,y,\theta)}(V)$ as in~\eqref{Psiwindow},
with a quadratic form on the fibers of $\cE$  over the local chart, determined by a smooth section 
$A$ of $T^*S \otimes T^*S$ that is symmetric and positive definite, we consider an
associated quadratic form
\begin{equation}\label{Qform}
 \cQ: \cE\oplus \cE^\vee \to \C, \ \ \ \ \cQ_{(x,y,\theta)}(W+i\eta):= W^t \, A_{(x,y)} \, W + 2i \langle  \eta, W \rangle_{(x,y,\theta)} - \eta^t\, \eta, 
\end{equation} 
where $\langle \eta, W \rangle$ is the duality pairing of $\cE$ and $\cE^\vee$, and $\eta^t\, \eta$
denotes the pairing with respect to the metric in $\cE^\vee$ 
determined by the metric on $S$.
We use the notation 
\begin{equation}\label{znotation}
\cQ(z):= \cQ\circ \cI^{-1}(z) \ \ \ \text{ and } \ \ \ V\bullet z:= V^t \frac{A_{(x,y)}}{2\pi} W +i \langle \eta, V \rangle\, .
\end{equation} 
We also define $\tilde\cQ: \cE\oplus \cE^\vee \to \C$ as 
\begin{equation}\label{tildeQ}
\tilde\cQ_{(x,y,\theta)}(W,\eta):= \frac{\pi}{2}\big( W^t \frac{A_{(x,y)}}{\pi} W +(\eta+\frac{\eta_{\theta}}{2\pi})^t (\eta+\frac{\eta_{\theta}}{2\pi})\,\big) . 
\end{equation}
We write $\tilde\cQ(z):=\tilde\cQ\circ \cI^{-1}(z)$.

\begin{definition}\label{BargmGeom}
The Bargmann transform of a function $f \in L^2(\cE,\C)$ is a function $\cB\, f:\cE\oplus\cE^\vee\to \C$
defined fiberwise by
\begin{equation}\label{eqBargGeom}
    (\cB\, f)|_{(\cE\oplus\cE^\vee)_{(x,y,\theta)}}(W,\eta):=\int_{\cE_{(x,y,\theta)}} \, f|_{\cE_{(x,y,\theta)}}(V)\, e^{2\pi V\bullet z-\pi V^t A_{(x,y)}V+\frac{\pi}{2}\cQ(z)}\, 
    d{\rm vol}_{(x,y,\theta)}(V)\, ,
\end{equation}
with the notation as in~\eqref{znotation} and with $d{\rm vol}_{(x,y,\theta)}(V)$ the volume 
form on the
fibers of $\cE$ determined by the Riemannian metric on $S$. 
\end{definition}

\begin{lemma}\label{GaborBargmannLem}
Consider the window function $\Psi_0$ as in~\eqref{Psiwindow}. The Gabor functions
\[
\rho(W,\eta)\Psi_0(V)=e^{2\pi i\langle \eta,V-W \rangle}\Psi_0(V-W)\, , \]
with $(W,\eta)\in\mathcal{E}\bigoplus\mathcal{E}^\vee$, satisfy 
\begin{equation}\label{eqGaborBargmann}
    |\langle f,\rho(W,\eta)\Psi_0 \rangle|=|\cB\, f(W-i(\eta+\frac{\eta_\theta}{2\pi}))|\, 
    e^{-\tilde\cQ(W,\eta)} \, .
\end{equation}
with $\tilde\cQ$ as in~\eqref{tildeQ}.
\end{lemma}

\begin{proof} 
We have
\begin{align*}
    \langle f,\rho(W,\eta)\Psi_0\rangle &=\int_{\cE_{(x,y,\theta)}}f(V)\,e^{-\pi(V-W)^t\frac{A}{\pi}(V-W)-2\pi i \langle \frac{\eta_\theta}{2\pi},V-W\rangle }e^{-2\pi i \langle \eta,V\rangle}\, d{\rm vol}(V)\\
    &=e^{2\pi i \langle \frac{\eta_\theta}{2\pi},W\rangle}\int_{\cE_{(x,y,\theta)}} f(V)\, e^{\pi V^t\frac{A}{\pi}V-2\pi V^t\frac{A}{\pi}W-\pi W^t\frac{A}{\pi}W}e^{-2\pi i\langle \frac{\eta_\theta}{2\pi}+\eta,V \rangle}\, d{\rm vol}(V)\\
    &= e^{2\pi i \langle \frac{\eta_\theta}{2\pi},W\rangle} e^{-i\pi \langle \eta+\frac{\eta_\theta}{2\pi}, W\rangle}e^{-\frac{\pi}{2}W^t\frac{A}{\pi}W+\frac{\pi}{2}(\eta+\frac{\eta_\theta}{2\pi})^t\cdot (\eta+\frac{\eta_\theta}{2\pi})}\cdot \\ 
    &\cdot \int_{\cE_{(x,y,\theta)}}f(V)\, e^{-2\pi V\bullet (W-i(\frac{\eta_\theta}{2\pi}+\eta))}e^{-\pi V^t\frac{A}{\pi}V}e^{-\frac{\pi}{2} \cQ(W-i(\eta+\frac{\eta_\theta}{2\pi}))}d{\rm vol}(V)\,
    \end{align*}
with $\cQ$ as in~\eqref{Qform} and $\tilde\cQ$ as in~\eqref{tildeQ}. 
\end{proof}

\begin{remark}\label{remBz}
Under the identification~\eqref{Imap} we write~\eqref{eqGaborBargmann} equivalently as
\begin{equation}\label{eqBcoeff}
 |\langle f,\rho(W,\eta)\Psi_0\rangle |=|\cB f(\overline{z})|e^{-\tilde\cQ(W,\eta)} \ \ \text{ for } \ 
z=W+i(\frac{\eta_\theta}{2\pi}+\eta) \, . 
\end{equation} 
\end{remark}

\begin{definition}\label{BFspaceE}
The global Bargmann-Fock space $\cF^2(\cE \oplus\cE^\vee)$ is the space of functions 
$F:\cE \oplus\cE^\vee \to \C$ such that 
$F|_{(\cE \oplus\cE^\vee)_{(x,y,\theta)}} \circ \cI^{-1}:\C^2\to \C$ is entire with 
\[ \| F \|^2_{\cF^2(\cE \oplus\cE^\vee)}=\int_M \int_{\C^2} 
\bigg|F|_{(\cE \oplus\cE^\vee)_{(x,y,\theta)}}\circ \cI^{-1}(z)\bigg|^2 \, e^{-2\tilde\cQ(z)}dz\, 
d{\rm vol}(x,y,\theta) < \infty \, . \]
The fiberwise Bargmann-Fock space $\cF^2(\cE \oplus\cE^\vee)_{(x,y,\theta)}$ is the space of functions 
$F: (\cE \oplus\cE^\vee)_{(x,y,\theta)} \to \C$ such that $F\circ \cI^{-1}:\C^2\to \C$ is entire,
with the norm
\[ \| F \|^2_{\cF^2(\cE \oplus\cE^\vee)_{(x,y,\theta)}} =\int_{\C^2} 
\bigg|F|_{(\cE \oplus\cE^\vee)_{(x,y,\theta)}}\circ \cI^{-1}(z)\bigg|^2 \, e^{-2\tilde\cQ(z)}dz < 
\infty \, . \]
\end{definition}

The space $\cF^2(\cE \oplus\cE^\vee)$ is a Hilbert space with the inner product 
\[\langle F,G\rangle_{\cF^2}:=\int_M \int_{\mathbb{C}^2}F|_{(\cE 
\oplus\cE^\vee)_{(x,y,\theta)}} \circ \cI^{-1}(z) \,\,\overline{G|_{(\cE 
\oplus\cE^\vee)_{(x,y,\theta)}} \circ \cI^{-1}(z)}\, e^{-2\tilde\cQ(z)}\, dz \, d{\rm vol}(x,y,\theta)\, 
. 
\]
Indeed, $\cF^2(\cE \oplus\cE^\vee)$ is the direct integral over $(M,d{\rm vol})$ of a family of Hilbert spaces
$\cF^2(\cE \oplus\cE^\vee)_{(x,y,\theta)}$, which are isomorphic, through the map $\cI$ 
of~\eqref{Imap} with
the Hilbert space $L^2(\C^2,e^{-2\tilde\cQ(z)}dz)$.

In this geometric setting we formulate the sampling condition in the following way. 

\begin{definition}\label{smoothsample}
Let $\Lambda$ be a bundle of lattices over $M$ where, over a local chart we have $\Lambda_{(x,y,\theta)}$ a lattice in
$(\cE\oplus\cE^\vee)_{(x,y,\theta)}$. The bundle $\Lambda$ satisfies the smooth 
sampling condition for $\cF^2(\cE \oplus\cE^\vee)$ if there are $\R^*_+$-valued smooth functions 
$C,C'$ on the local charts of $M$, such that, for all $(x,y,\theta)$ in a local chart of $M$ and for all 
$F\in \cF^2(\cE \oplus\cE^\vee)_{(x,y,\theta)}$, the estimates
\begin{equation}\label{Msample}
C_\mu \cdot \| F \|^2_{\cF^2(\cE \oplus\cE^\vee)_\mu} \leq \sum_{(W,\eta)\in \Lambda_\mu} \bigg| F|_{(\cE \oplus\cE^\vee)_\mu} \bigg|^2 e^{-2\tilde\cQ_\mu(W,\eta)} \leq C'_\mu \cdot \| F \|^2_{\cF^2(\cE \oplus\cE^\vee)_\mu} 
\end{equation}
are satisfied, for $\mu=(x,y,\theta)$ in a local chart of $M$, and with $\tilde\cQ$ as 
in~\eqref{tildeQ}. 
\end{definition}

\begin{lemma}\label{lemEframe}
For any $(x,y,\theta)$ in a local chart of $M$,
the Bargmann transform $\cB$ of~\eqref{eqBargGeom} is a bijection from 
$L^2(\cE_{(x,y,\theta)})$ to $\cF^2{(\cE \oplus\cE^\vee)}_{(x,y,\theta)}$, with
\begin{equation}\label{metricEquality}
    \| \cB \, f \|_{\cF{(\cE\oplus\cE^\vee)}_{(x,y,\theta)}}= 
    K_{(x,y)}\cdot \| f \|_{L^2(\cE_{(x,y,\theta)}) } \, ,
\end{equation}
for a smooth $\R_+^*$-valued function $K$ over the local charts $U$ of $S$. 
Moreover, $\cG(\Psi_0,\Lambda)$ is a frame for $L^2(\cE_{(x,y,\theta)})$ if and only if 
$\overline{\Lambda}+i\frac{\eta_\theta}{2\pi}$ is a set of sampling for 
$\cF{(\cE\oplus\cE^\vee)}_{(x,y,\theta)}$.
\end{lemma}

\begin{proof}
For the window function $\Psi_0$ as in~\eqref{Psiwindow}, we have
\[ \|\Psi_0 \|^2_{L^2(\cE_{(x,y,\theta)})} =\int_{\cE_{(x,y,\theta)}}  | \Psi_0(V) |^2\, dV =
 \int_{\cE_{(x,y,\theta)} } e^{-2 V^t A_{(x,y)} V} \, dV =  \frac{\pi}{2\sqrt{\det(A_{(x,y)})} }  \, . \] 
 as a standard Gaussian integral in $2$-dimensions. Because we assumed that the
matrices $A_{(x,y)}$ in the window function $\Psi_0$ of~\eqref{Psiwindow} have spectrum
bounded away from zero, and that $S$ is compact, the quantity
\[ K_{(x,y)} := \frac{\pi} { 2 \sqrt{ \det(A_{(x,y)})} }   \]
determines a smooth real valued function $K: S\to \R$ with a strictly positive minimum and 
a bounded maximum.
Moreover, by Theorem~3.2.1 and Corollary~3.2.2 of~\cite{Groch}, the orthogonality
relation 
\[
 \langle V_{\phi_1} f_1, V_{\phi_2} f_2\rangle_{L^2(\R^{2n})} = \langle f_1,f_2 
 \rangle_{L^2(\R^n)} \cdot \overline{\langle \phi_1,\phi_2 \rangle_{L^2(\R^n)}} \]
for the short time Fourier transform 
\[ V_\phi f (x,\omega)=\int_{\R^n} f(t)\, \overline{\phi(t-x)} \, e^{-2\pi i t\cdot \omega} \, dt \, , \ 
\text{ with }\ (x,\omega)\in \R^{2n}\, , \]
gives the identity
\[ \|\langle f,\rho(W,\eta)\Psi_0\rangle\|_{L^2(\cE_{(x,y,\theta)})}= \| f 
\|_{L^2(\cE_{(x,y,\theta)})}  
\cdot \| \Psi_0 \|_{L^2(\cE_{(x,y,\theta)})} \, .\]
Moreover, by~\eqref{eqBcoeff} we have, for $z=\cI(W,\eta)$,

\begin{align*}
\|\langle f,\rho(W,\eta)\Psi_0\rangle\|_{L^2(\cE_{(x,y,\theta)})} & = 
\int_{{(\cE\oplus\cE^\vee)}_{(x,y,\theta)}} | \langle f,\rho(W,\eta)\Psi_0\rangle |^2 d{\rm 
vol}(W,\eta) \\
&  = \int_{\C^2} |\cB\, f(\bar z) |^2 e^{-2\tilde\cQ(z)} \, dz\, . 
 \end{align*}
Injectivity then follows, while surjectivity follows by the same argument showing
the density of $\cB(L^2(\R^n))\subset \cF^2_n$ in the proof of Theorem~3.4.3 
of~\cite{Groch}, 
applied pointwise in $(x,y,\theta)\in M$.

The Gabor system $\cG(\Psi_0,\Lambda)$ satisfies the smooth frame
condition of Definition~\ref{smoothGaborDef} if 
there are smooth functions $C_{(x,y,\theta)}, C'_{(x,y,\theta)} >0$ on the local charts of $M$ such that
\[ C_{(x,y,\theta)} \, \| f \|^2_{L^2(\cE_{(x,y,\theta)})} \leq \sum_{\lambda=(W,\eta)\in \Lambda}
|\langle f , \rho(\lambda) \Psi_0 \rangle |^2 \leq C'_{(x,y,\theta)} \| f 
\|^2_{L^2(\cE_{(x,y,\theta)})}\, . \]
By~\eqref{eqBcoeff} we see that this is equivalent to the smooth sampling condition of 
Definition~\ref{smoothsample} for $\overline{\Lambda}+i\frac{\eta_\theta}{2\pi}$. 
\end{proof}

\begin{proposition}\label{bFrameProp}
With the scaling by the function $b=b_M(x,y,\theta)$ of~\eqref{afunctradii}, the Gabor system
$\cG(\Psi_0,\Lambda_{b,\alpha,J}\oplus
\Lambda_{\alpha,J}^\vee)$ satisfies the frame condition.
\end{proposition}

\begin{proof} We write here the window function $\Psi_0$ as $\Psi_0^A$ to emphasize the
dependence on the quadratic from $A=A_{(x,y)}$.
Let $f: \cE \to \R$ be a signal, with $f|_{\cE_{(x,y,\theta)}}\in L^2(\cE_{(x,y,\theta)})$.
We have
\begin{multline*}
	\sum_{\lambda\in \Lambda_{b,\alpha,J}\oplus
\Lambda_{\alpha,J}^\vee} \bigg| \langle f,\rho(\lambda)\Psi^A_0\rangle \bigg|^2  \\
 = \sum_{(n,m)\in \Z^2\times\Z^2}\bigg| \int_{\cE_{(x,y,\theta)}}f(V)e^{2\pi i m\cdot 
 V}e^{{(V-bn)}^t A_{(x,y,\theta)}(V-bn)+i\langle \eta_\theta,V-bn\rangle }d{\rm 
 vol}_{(x,y,\theta)}(V)\bigg|^2\, . 
\end{multline*}
With a change of variables $V=\sqrt{b}\, U$ and correspondingly changing the quadratic form $A$ to $b A$,
we rewrite the above as
\begin{multline*}
\sum_{(n,m)\in \Z^2\times\Z^2} b^2\, \bigg| \int_{\cE_{(x,y,\theta)}} f(\sqrt{b}U) 
\overline{M_{\sqrt{b}m}T_{\sqrt{b}n}e^{-U(bA_{(x,y,\theta)})U^T-i\langle \eta_\theta\sqrt{b},U \rangle}}
d{\rm vol}_{(x,y,\theta)}(V)\bigg|^2  \\
 =  b^2\sum_{\tilde{\lambda}\in \sqrt{b}\Lambda_{\alpha,J}\oplus
\sqrt{b}\Lambda_{\alpha,J}^\vee}\bigg| \langle f_{\sqrt{b}},\rho(\tilde{\lambda})\Psi^{bA}_0 
\rangle \bigg|^2\, , 
\end{multline*}
 where $f_{\sqrt{b}}(V)=f(\sqrt{b}V)$. Therefore, the Gabor system $\cG(\Psi_0^A,\Lambda_{b,\alpha,J}\oplus
\Lambda_{\alpha,J}^\vee)$ is a frame for $L^2(\cE_{(x,y,\theta)})$ if and only if 
$\cG(\Psi_0^{bA},\Lambda_{\sqrt{b},\alpha,J}\oplus
\Lambda_{\sqrt{b},\alpha,J}^\vee)$ is a frame for $L^2(\cE_{(x,y,\theta)})$. 
Moreover, by Lemma~\ref{lemEframe}, we know that $\cG(\Psi_0^{bA}, \Lambda_{\sqrt{b},\alpha,J}\oplus
\Lambda_{\sqrt{b},\alpha,J}^\vee)$ is a frame for $L^2(\cE_{(x,y,\theta)})$ if and only if the uniformly 
discrete set $(\sqrt{b}\mathbb{Z}^2+i\sqrt{b}\mathbb{Z}^2)+i\frac{\eta_{\theta}}{2\pi}$
is a set of sampling for $\cF{(\cE\oplus\cE^\vee)}_{(x,y,\theta)}$. Finally, by~\cite{Groch2}, 
$(\sqrt{b}\mathbb{Z}^2+i\sqrt{b}\mathbb{Z}^2)+i\frac{\eta_{\theta}}{2\pi}$
is a set of sampling if and only if the complex lattice $T(\Z^2+i\Z^2)$ is a set of sampling, for 
the matrix 
\[
T=\begin{pmatrix}
\sqrt{b}&0\\
0&\sqrt{b}
\end{pmatrix}\, . \]
By Proposition~11 of~\cite{Groch2}, the latter condition is satisfied if and only if $\sqrt{b}<1$, 
which we know
is the case by Lemma~\ref{scaledlatt}. 
\end{proof}

\section{Gabor frames: symplectization and contactization}\label{Gabor5DSec}

As in Section~\ref{SymplContSec}, we consider the contactization $\cC\cS(M)$ of the 
symplectization $\cS(M)$ 
of the manifold of contact elements $M=\bS_w(T^*S)$ of a surface $S$. 
This model is motivated by the goal of describing visual perception based on neurons sensitive not only to 
orientation, but also to frequency and phase, with the frequency-phase and the position-orientation uncertainty
minimized by the Gabor functions profiles. From the point of view of this model, it is worth pointing out that, 
although higher dimensional, the $5$-dimensional contact manifold $\cC\cS(M)$ is completely determined
by the contact $3$-manifold $M$ with no additional independent choices, being just the contactization of the symplectization.

Note that, while the contact structure of $\cC\cS(M)$ is the natural extension of the
contact structure of $M$, this does not directly imply that modelling the visual cortex 
requires an increasing family of contact structures to account for different
families of cells sensitive to different features, as different features may be
described by the same geometry.

Given local charts on $M$ with the choice of local basis 
\begin{equation}\label{xibasis}
X_\theta =\partial_\theta, \ \ \  R_{\alpha_J} =- w^{-1}\sin(\theta) \partial_{x} +w^{-1} 
\cos(\theta) \partial_y
\end{equation}
for the contact planes $\xi$ of the contact structure $\alpha$ on $M$ and
the Reeb field $R_\alpha=w^{-1} \cos(\theta) \partial_{x} + w^{-1} \sin(\theta) \partial_y$, we 
obtain
a basis of the contact hyperplane distribution of the five-dimensional contact manifold
$(T^*S_0\times S^1, \tilde\alpha)$, in the corresponding local charts, given by 
\begin{align*}
 \{ X_\theta, R_{\alpha_J}, R_{\phi,\alpha}, X_w  \},\\
 \text{ with } \quad R_{\phi,\alpha}:=\partial_\phi + R_\alpha \ \ \text{ and } \ \  X_w :=\partial_w\, 
 . 
\end{align*}
In the case of the twisted contact structure $\alpha_J$, with the choice of basis
\begin{equation}\label{xiJbasis}
X_\theta =\partial_\theta, \ \ \  R_\alpha=w^{-1} \cos(\theta) \partial_{x} + w^{-1} \sin(\theta) 
\partial_y
\end{equation}
for the contact plane distribution $\xi_J$, 
and the Reeb vector field $R_{\alpha,J} =- w^{-1}\sin(\theta) \partial_{x} +w^{-1} \cos(\theta) 
\partial_y$,
we similarly obtain a basis for the contact hyperplanes $\tilde\xi_J$ given by
\begin{align}\label{basisxiJ}
 \{ X_\theta, R_\alpha, R_{\phi,\alpha,J}, X_w \}, \\
 \nonumber
 R_{\phi,\alpha,J}:=\partial_\phi + R_{\alpha,J} \, . 
\end{align} 
The bundle $\cE$ of signal planes on $M$ determines the following bundles on
the symplectization $\cS(M)$ and the contactization $\cC\cS(M)$. 
\begin{definition}\label{EhatEtildeE}
Let $\hat\cE$ denote the pull-back of the bundle $\cE$ of signal planes to 
$T^*S_0\simeq M\times \R^*_+$ via the projection to $M$, and 
let $\tilde\cE$ denote the vector bundle over $\cC\cS(M)$ given by 
$\hat\cE\boxplus TS^1=\pi_{T^*S_0}^* \hat\cE \oplus \pi_{S^1}^* TS^1$,
with pull-backs taken with respect to the two projections of $\cC\cS(M)=T^*S_0\times S^1$ 
on the two factors.
\end{definition}

The signals in this setting will be functions $I: \tilde\cE\to \R$. The vector bundle $\tilde\cE$ on $\cC\cS(M)$
is a rank $3$ real vector bundle over a $5$-dimensional manifold.

\begin{remark}\label{basisEtilde}
A basis of sections for $\tilde\cE$ over a local chart is obtained
by taking the vectors $\{ R_\alpha, R_{\alpha_J}, \partial_\phi \}$.
There are two other choices of basis directly determined by the contact forms $\tilde\alpha$ and $\tilde\alpha_J$, namely
$\{ R_\alpha, R_{\phi,\alpha,J}, R_{\tilde\alpha_J} \}$, where the first two vectors span the 
intersection
$\tilde\cE\cap \tilde\xi$ of the contact hyperplane distribution with the bundle $\tilde\cE$ and the last
vector is the Reeb field of $\tilde\alpha_J$, or $\{ R_{\alpha_J}, R_{\phi,\alpha}, R_{\tilde\alpha} \}$, 
with the first two vector fields spanning $\tilde\cE\cap \tilde\xi_J$ and the third the Reeb field
of $\tilde\alpha$. The first basis has the advantage of a providing consistent choices of basis for both
$\cE$ and $\tilde\cE$. 
\end{remark}

\begin{lemma}\label{lem5Dwindow}
In a local chart of $\cS(M)$ with coordinates $(x,y,w,\theta)$,
the window function $\Psi_0$ as in~\eqref{Psiwindow} extends to a window function
on $\hat\cE$ given by
\begin{equation}\label{hatPsi0}
 \hat\Psi_{0, (x,y,w,\theta)}(V) =\exp\left(-V^t A_{(x,y)} V - i \langle \eta_{(w,\theta)}, V\rangle_{(x,y)} \right), 
\end{equation} 
with $\eta_{(w,\theta)}=(w\cos(\theta), w\sin(\theta))$.
\end{lemma}

\begin{proof} The window function $\Psi_0$ as in~\eqref{Psiwindow} is obtained as 
restriction
to $TS\oplus \bS(T^*S)$ of a window function $\Phi_0$ on $TS\oplus T^*S$ defined
as in~\eqref{Phi0S} in Definition~\ref{Psi0def}. By identifying $\cS(M)=T^*S_0$ and
$\hat\cE$ with the pull-back of $TS$ to $\cS(M)$, we see that $\Phi_0$ induces a
window function $\hat\Psi_0$ on $\hat\cE$ of the form~\eqref{hatPsi0}.
\end{proof}

We further extend the window function~\eqref{hatPsi0} to $\tilde\cE$ so as to
obtain a window function that is a modified form of the function considered in the
model of~\cite{BaspSartiCitti}.

\begin{definition}\label{tildePsi0def}
In a local chart of $\cC\cS(M)$ with coordinates $(x,y,w,\theta,\phi)$, 
window functions on $\tilde\cE$ extending the window function~\eqref{hatPsi0} are
functions on $\tilde\cE$ of the form
\begin{equation}\label{tildePhi0}
\tilde\Psi_{0,(x,y,w,\theta,\phi),\zeta_0}(V,\upsilon) =\exp \left( - V^t A_{(x,y)} V -i \langle \eta_{(w,\theta)}, V \rangle_{(x,y,w,\theta)}-\kappa_\phi^2 \upsilon^2  -i \langle \zeta_0, \upsilon \rangle_\phi \right)\, ,
\end{equation}
for $\eta_{(w,\theta)}$ as in~\eqref{hatPsi0}, and with $\zeta_0 \in T^*_\phi S^1$ and 
$\upsilon \in T_\phi S^1$.
The two-dimensional
Gabor systems of the form $\{ \rho(W,\eta) \Psi_0 |_{\cE_{(x,y,\theta)}}  \}$ are then replaced 
by a three-dimensional system of the form
\begin{equation}\label{Gabor3D}
 \rho(W,\eta,\nu,\zeta) \tilde\Psi_0 |_{\tilde\cE_{(x,y,w,\theta,\phi)}} (V,\upsilon) \, , 
\end{equation} 
with $(W,\eta,\nu,\zeta)\in {(\tilde\cE\oplus\tilde\cE^\vee)}_{(x,y,w,\theta,\phi)}$. 
\end{definition}

In the setting of~\cite{BaspSartiCitti}, the additional variables $\phi\in S^1$ (with
its linearization  $\upsilon \in T_\phi S^1$) and the dual variable $\zeta \in T^*_\phi S^1$,
which we view here as part of the bundle $\tilde\cE$ over the contact manifold $\cC\cS(M)$, 
represent a model of phase and velocity of spatial wave propagation. The window function
$\tilde\Psi_0$ that we consider here differs from the function considered 
in~\cite{BaspSartiCitti},
which does not have the Gaussian term in the $\upsilon \in T_\phi S^1$ variable. While
they consider the limit case where $\kappa_\phi=0$, we argue here that one needs this
additional term to be non-zero (though possibly small) in order to have good signal analysis
properties for the associated Gabor system, in the presence of these additional variables. 
The Gaussian term in $\upsilon$ can in principle be replaced by another rapid decay function,
however, it seems more natural to use a Gaussian term, like we have for the variables in $\cE$,
in order to maintain a similar structure for all the variables of $\tilde\cE$.  We will return to
discuss the case $\kappa_\phi=0$ of~\cite{BaspSartiCitti} in Section~\ref{Gelfand3Sec}.

Note that the goal of the model of~\cite{BaspSartiCitti} is different, as they apply the
Gabor transform to signals that are independent of the frequency and phase variables,
so that the problem outlined above with the frame condition does not arise. 
It is only when the signal analysis is performed on the larger space given by
the $3$-dimensional linear fibers of the bundle $\tilde \cE$, rather than on the
$2$-dimensional bundle $\hat\cE$, that one needs to modify the window
function as described above.

Let $\tilde\cE^\vee$ denote the dual bundle of $\tilde\cE$, with the choice of local basis
$\{ R_\alpha, R_{\alpha,J}, \partial_\phi \}$ for $\tilde\cE$ and the dual local
basis $\{ \alpha, \alpha_J, d\phi \}$. This determines bundles of framed lattices over the local charts of $\cC\cS(M)$
\begin{equation}\label{tildeLatt}
 \tilde\Lambda = \Z\, R_\alpha + \Z\, R_{\alpha_J} + \Z\, \partial_\phi = \Lambda_{\alpha,J}\oplus L, \ \ \ \
 \tilde\Lambda^\vee = \Z \,\alpha + \Z\, \alpha_J  + \Z \, d\phi = \Lambda^\vee_{\alpha,J}\oplus L^\vee \, ,
\end{equation}
with $\Lambda_{\alpha,J}$ and $\Lambda^\vee_{\alpha,J}$ the bundles of framed lattices 
of~\eqref{latLambda} and~\eqref{latLambdavee}.

We consider the bundle of framed lattices 
$\tilde\Lambda\oplus \tilde\Lambda^\vee$, which has the property that, in a local chart, the fibers
\[
{(\tilde\Lambda\oplus \tilde\Lambda^\vee)}_{(x,y,w,\theta,\phi)} \subset 
{(\tilde\cE \oplus \tilde\cE^\vee)}_{(x,y,w,\theta,\phi)} 
\]
are lattices in the fibers of the $3$-plane bundle $\tilde\cE\oplus \tilde\cE^\vee$ over the
$5$-dimensional contact manifold $\cC\cS(M)$.

The window function $\tilde\Psi_0$ and the bundle of framed lattices 
$\tilde\Lambda\oplus \tilde\Lambda^\vee$ determine a Gabor system
\begin{equation}\label{Gabor3DLambda}
\cG(\tilde\Psi_0, \tilde\Lambda\oplus \tilde\Lambda^{\vee})=\left\{ \rho(\lambda) \tilde\Psi_0 
|_{\tilde\cE_{(x,y,w,\theta,\phi)}} \, \bigg|\, \lambda=(W,\eta,\nu,\zeta) \in {(\tilde\Lambda\oplus 
\tilde\Lambda^{\vee})}_{(x,y,w,\theta,\phi)}  \right\} \, .
\end{equation}

As in the case of the bundle of framed lattices $\Lambda_{\alpha,J}\oplus \Lambda^\vee_{\alpha,J}$
we consider a scaling of the lattices in the fibers of $\tilde\Lambda\oplus \tilde\Lambda^\vee$, for
the same reasons discussed in Section~\ref{scaleLSec}. We define $R_{\max}>0$ as in 
Section~\ref{scaleLSec}. For the $TS^1$ direction of $\tilde\cE$, the injectivity radius 
$R^{S^1}_{inj}$ 
is constant and equal to half the length of the $S^1$ circle.  Thus, we take, as in 
Section~\ref{scaleLSec},
a scaling factor of the form $\gamma=R^{S^1}_{inj}/R_{\max}$. As discussed in 
Section~\ref{scaleLSec}
we can assume that in our model $R_{\max} > R^{S^1}_{inj}$ so that $\gamma <1$. We then
consider the bundle of framed lattices determined by this choice of scaling on $L$ and the previous
choice of scaling on $\Lambda_{\alpha,J}$.

\begin{definition}\label{scale3Dlattdef}
Let $\tilde\Lambda_{b,\gamma}\oplus \tilde\Lambda^\vee$ be the bundle of framed lattices of the form
\begin{equation}\label{scale3Dlat}
 \tilde\Lambda_{b,\gamma}\oplus \tilde\Lambda^\vee = 
\Lambda_{b,\alpha,J}\oplus L_\lambda \oplus \Lambda_{\alpha,J}^\vee \oplus L^\vee\, ,
\end{equation}
where $\Lambda_{b,\alpha,J}$ is the scaled lattice of~\eqref{aLambda} with $b=b_M$ the
function of~\eqref{afunctradii}, while $L_\lambda= \lambda \, L$ for the constant 
$\lambda=R^{S^1}_{inj}/R_{\max}$ as above. This has associated Gabor system
\begin{equation}\label{scaleGabor3DLambda}
\cG(\tilde\Psi_0, \tilde\Lambda_{b,\gamma}\oplus \tilde\Lambda^\vee)=\left\{ \rho(\lambda) \tilde\Psi_0  \, \bigg|\, \lambda=(W,\eta,\nu,\zeta) \in \tilde\Lambda_{b,\gamma}\oplus \tilde\Lambda^\vee  \right\} \, ,
\end{equation}
where for simplicity of notation we have suppressed the explicit indication of the fibers of
$\tilde\cE\oplus\tilde\cE^\vee$ as in~\eqref{Gabor3DLambda}.
\end{definition}

We then have the following result about the Gabor frame condition for the Gabor
systems~\eqref{Gabor3DLambda} and~\eqref{scaleGabor3DLambda}.

\begin{proposition}\label{prop3Dframes}
The Gabor system $\cG(\tilde\Psi_0, \tilde\Lambda\oplus \tilde\Lambda^\vee)$ 
of~\eqref{Gabor3DLambda} 
is not a frame. The Gabor system $\cG(\tilde\Psi_0, \tilde\Lambda_{b,\gamma}\oplus \tilde\Lambda^\vee)$
of~\eqref{scaleGabor3DLambda} is a frame. 
\end{proposition}

\begin{proof} By construction the Gabor systems with window function $\tilde\Psi_0$ and 
lattice
$\tilde\Lambda \oplus \tilde\Lambda^\vee$ or $\tilde\Lambda_{b,\gamma}\oplus \tilde\Lambda^\vee$
split as a product of a $2$-dimensional system $\cG(\Psi_0,\Lambda_{\alpha,J}\oplus \Lambda_{\alpha,J}^\vee)$
or $\cG(\Psi_0,\Lambda_{b,\alpha,J}\oplus \Lambda_{\alpha,J}^\vee)$ and a $1$-dimensional 
Gabor
system $\cG(\psi_0, L\oplus L^\vee)$ or $\cG(\psi_0, L_\lambda \oplus L^\vee)$, where 
\[
 \psi_{0,\phi}(\upsilon) = \exp(-\kappa_\phi \upsilon^2 -i \langle \zeta_0,\upsilon\rangle_\phi )=
\exp(-\kappa_\phi \upsilon^2 -i \zeta_0 \upsilon )\, . 
\]
Thus, the frame condition
for $\cG(\tilde\Psi_0, \tilde\Lambda\oplus \tilde\Lambda^\vee)$ holds if and only if it holds
for both $\cG(\Psi_0,\Lambda_{\alpha,J}\oplus \Lambda_{\alpha,J}^\vee)$ and $\cG(\psi_0, L\oplus L^\vee)$
and similarly the frame condition for $\cG(\tilde\Psi_0, \tilde\Lambda_{b,\gamma} \oplus \tilde\Lambda^\vee)$
holds if and only if it holds for both $\cG(\Psi_0,\Lambda_{b,\alpha,J}\oplus \Lambda_{\alpha,J}^\vee)$ and
$\cG(\psi_0, L_\lambda \oplus L^\vee)$.
For the $1$-dimensional systems with a rapid decay function as window function, the frame condition
holds if and only if the lower Beurling density $D^-$ of the lattice is strictly greater than one. 
For the lattice $L\oplus L^\vee$ this condition is not satisfies (see Proposition~\ref{noframes})
so the Gabor system is not a frame, while for the lattice 
$L_\lambda\oplus L^\vee$ is satisfied since $\gamma<1$ (see Proposition~\ref{yesframes}).
Thus, in the case of the Gabor system $\cG(\tilde\Psi_0, \tilde\Lambda_{b,\gamma}\oplus \tilde\Lambda^\vee)$
of~\eqref{scaleGabor3DLambda} the question is reduced to the question of whether the 
$2$-dimensional
system $\cG(\Psi_0,\Lambda_{b,\alpha,J}\oplus \Lambda_{\alpha,J}^\vee)$ is a frame. We know this system
is indeed a frame by Proposition~\ref{bFrameProp}. 
\end{proof}

\subsection{Gelfand triples and Gabor frames}\label{Gelfand3Sec}

We return here to discuss the case of the profiles considered in~\cite{BaspSartiCitti},
with the term $\kappa_\phi=0$. As mentioned above, the function $\tilde\Psi_0 |_{\kappa_\phi=0}$
is not a window function for a Gabor system in the usual sense, as it is not of rapid decay
(and not even $L^2$) along the fibers of $\tilde\cE$. However, we can still interpret it as
a tempered distribution on the fibers of $\tilde\cE$. Thus, one can at least ask the question of
whether this window function defines Gabor frames in a distributional sense. 
To formulate Gabor systems  in such a setting, it is convenient to consider the formalism of Gelfand
triples (also known as rigged Hilbert spaces,~\cite{GeVi}).

We consider here the same setting as in~\cite{TTT, Tsch} for distributional frames,
with Gelfand triples given by 
\[ \cS(\tilde\cE_{(x,y,w,\theta,\phi)}) \hookrightarrow L^2(\tilde\cE_{(x,y,w,\theta,\phi)}, d{\rm 
vol}_{(x,y,w,\theta,\phi)}) \hookrightarrow \cS^\prime(\tilde\cE_{(x,y,w,\theta,\phi)})\, , \]
where the space $\cS$ of tempered distributions is densely and continuously embedded 
in the $L^2$-Hilbert space, which is densely and continuously embedded in the dual space 
$\cS^\prime$ of distributions. The pairing $\langle f, g \rangle$ of distributions $f\in \cS^\prime$
and test functions $g\in \cS$ extends the Hilbert space inner product. We write the above
triples for simplicity of notation in the form
\[ \cS(\tilde\cE)\hookrightarrow L^2(\tilde\cE) \hookrightarrow \cS^\prime(\tilde\cE)\, . \]

\begin{definition}\label{DistrFramesDef}
A distributional Gabor system $\cG(\tilde\Phi_0,\tilde\Lambda)$ on $\tilde\cE$ is given by a 
window generalized-function $\tilde\Phi_0\in \cS'(\tilde\cE)$ and a bundle of lattices $\tilde\Lambda$ with
\[ \cG(\tilde\Phi_0,\tilde\Lambda)=\{ \rho(\lambda) \tilde\Phi_0 \,|\, \lambda\in \tilde\Lambda \} 
\subset \cS'(\tilde\cE)\, . \]
The distributional Gabor system $\cG(\tilde\Psi_0,\tilde\Lambda)$ is a distributional frame for the
bundle $\tilde\cE$ on $\cC\cS(M)$ if there are bounded smooth functions $C,C': \cC\cS(M)\to \R^*_+$  
with strictly positive $\inf_{\cC\cS(M)} C$ and $\inf_{\cC\cS(M)} C'$, such that, for all $f\in \cS(\tilde\cE)$
\[ C_{(x,y,w,\theta,\phi)} \, \| f \|^2_{L^2(\tilde\cE_{(x,y,w,\theta,\phi)})} \leq 
\sum_{\lambda\in \tilde\Lambda_{(x,y,w,\theta,\phi)}} |\langle \rho(\lambda)\tilde\Phi_0, f \rangle|^2
\leq C'_{(x,y,w,\theta,\phi)}\, \| f \|^2_{L^2(\tilde\cE_{(x,y,w,\theta,\phi)})} \, . \]
\end{definition}

\begin{lemma}\label{DistrGaborLem}
Let $\tilde\Phi_0=\tilde\Psi_0 |_{\kappa_\phi=0}$, with $\tilde\Psi_0$ as in~\eqref{tildePhi0}.
The systems $\cG(\tilde\Phi_0,\tilde\Lambda\oplus \tilde\Lambda^\vee)$ and 
$\cG(\tilde\Phi_0,\tilde\Lambda_{b,\gamma}\oplus \tilde\Lambda^\vee)$ with the lattices
as in Definition~\ref{scale3Dlattdef}, are distributional Gabor systems that decompose
into a product of a $2$-dimensional ordinary Gabor system given by $\cG(\Psi_0,\Lambda_{\alpha,J}\oplus
\Lambda_{\alpha,J}^\vee)$ or $\cG(\Psi_0,\Lambda_{b,\alpha,J}\oplus
\Lambda_{\alpha,J}^\vee)$, respectively, and a $1$-dimensional distributional Gabor system
of the form $\cG(\phi_0,L\oplus L^\vee)$ or $\cG(\phi_0,L_\gamma\oplus L^\vee)$, respectively,
with window generalized-function $\phi_0(\upsilon)=\exp(-i\zeta_0 \upsilon) \in \cS^\prime(\R)$.
The distributional Gabor system $\cG(\tilde\Phi_0,\tilde\Lambda\oplus \tilde\Lambda^\vee)$ does
not satisfy the distributional Gabor frame condition. The distributional Gabor system 
$\cG(\tilde\Phi_0,\tilde\Lambda_{b,\gamma}\oplus \tilde\Lambda^\vee)$ satisfies the distributional
Gabor frame condition if and only if the $1$-dimensional distributional Gabor system 
$\cG(\phi_0,L_\gamma\oplus L^\vee)$ satisfies the distributional frame condition.
\end{lemma}

\begin{proof} 
We view $\tilde\Phi_0$ as the distribution in $\cS'(\tilde\cE)$ that acts on test 
functions $f \in \cS(\tilde\cE)$ as
\[ \langle f,\tilde\Phi_0 \rangle_{(x,y,w,\theta,\phi)} = \int_{\tilde\cE_{(x,y,w,\theta,\phi)}} 
\overline{\tilde\Phi_0 |_{\tilde\cE_{(x,y,w,\theta,\phi)}}(V,\upsilon)} \, f(V,\upsilon) \, d{\rm 
vol}_{(x,y,w,\theta,\phi)}(V,\upsilon)\, .  \]
As in Proposition~\ref{prop3Dframes} we see that the distributions $\rho(\lambda)\tilde\Phi_0$ are
products of a function $\rho(\lambda')\Psi_0 \in \cS(\tilde\cE)$ and a distribution $\rho(\lambda'')\phi_0$
in $\cS^\prime(\tilde\cE)$, with $\lambda=(\lambda',\lambda'')$ for $\lambda\in \tilde\Lambda\oplus \tilde\Lambda^\vee$ and $\lambda'\in \Lambda_{\alpha,J}\oplus \Lambda_{\alpha,J}^\vee$ and
$\lambda''\in L\oplus L^\vee$ (and similarly for the scaled versions of the lattices). 
Since these Gabor systems decouple, the distributional frame condition becomes equivalent
to the ordinary frame condition for the part that is an ordinary frame and the distributional frame
condition for the part that is a distributional frame. Thus,
the distributional Gabor systems $\cG(\tilde\Phi_0,\tilde\Lambda\oplus \tilde\Lambda^\vee)$ and 
$\cG(\tilde\Phi_0,\tilde\Lambda_{b,\gamma}\oplus \tilde\Lambda^\vee)$ are distributional
Gabor frames if and only if the respective $2$-dimensional ordinary Gabor systems are ordinary 
frames and the respective $1$-dimensional distributional Gabor systems are distributional frames.
In the first case we know that the frame condition already fails at the level of the $2$-dimensional 
ordinary Gabor system. In the second case the $2$-dimensional system satisfies the usual
frame condition by Proposition~\ref{bFrameProp}, hence the question reduces to whether the
$1$-dimensional distributional system $\cG(\phi_0,L_\gamma\oplus L^\vee)$ satisfies the 
distributional frame condition.
\end{proof}

The following statement shows that, even when interpreted in this distributional setting
the Gabor system generated by the window function $\tilde\Phi_0$ as in~\cite{BaspSartiCitti}
does not give rise to frames, hence it does not allow for good signal analysis.

\begin{proposition}\label{noframe5D}
The distributional Gabor system $\cG(\tilde\Phi_0,\tilde\Lambda_{b,\gamma}\oplus \tilde\Lambda^\vee)$
does not satisfy the distributional frame condition.
\end{proposition}

\begin{proof} By Lemma~\ref{DistrGaborLem} we can equivalently focus on the question of whether the
one-dimensional distributional Gabor system $\cG(\phi_0,\gamma\Z+\Z)$ satisfies the 
distributional frame condition. Given a signal $f\in \cS(\R)$, we have, for $\lambda=(\gamma n,m)$
and $\phi_0(t)=e^{-i \zeta_0 t}$, 
\begin{align*}
	\langle f, \rho(\lambda)\phi_0 \rangle &= \int_\R e^{-2\pi i m t} f(t) e^{i \zeta_0 (t-\gamma n)} 
	\, dt \\
	&= e^{-i \zeta_0 \gamma n }\int_\R e^{-2\pi i (m-\frac{\zeta_0}{2\pi}) t} f(t) \, dt = 
	e^{-i \zeta_0 \gamma n } \hat f (m-\frac{\zeta_0}{2\pi}) \, . 
\end{align*}
Note that when we take $| \langle f, \rho(\lambda)\phi_0 \rangle |^2$ the dependence on $n$
disappears entirely so the sum over the lattice is always divergent.
\end{proof}

\begin{remark}\label{GroupHorVF} 
The window function $\tilde\Psi_0 |_{\kappa_\phi=0}$ in~\cite{BaspSartiCitti}
is chosen so that the Lie group and Lie algebra structure underlying receptive profiles 
of this form (see~\cite{Pet,SaCiPe})
determines horizontal vector fields given by the basis $\{ X_\theta, R_\alpha, R_{\phi,\alpha,J}, X_w \}$
of~\eqref{basisxiJ} of the contact hyperplanes $\tilde\xi_J$. However, if we replace this
choice of window with our window $\tilde\Psi_0$ where $\kappa_\phi\neq 0$, the same Lie group
of transformations acts on these types of profiles generating the same horizontal vector fields.
Note that also the original goal of~\cite{BaspSartiCitti} of describing receptive profiles of 
neurons sensitive to 
frequency and phase variables, with the frequency-phase uncertainty minimized is already 
satisfied by the 
Gabor system generated by our proposed window function $\tilde\Psi_0$, without the need to impose 
$\kappa_\phi=0$. \end{remark}



\providecommand{\bysame}{\leavevmode\hbox to3em{\hrulefill}\thinspace}


\ACKNO{This work was supported by NSF grants DMS-1707882 and DMS-2104330, 
NSERC grants RGPIN-2018-04937 and RGPAS-2018-522593, FQXi grants FQXi-RFP-1-804 and 
FQXI-RFP-CPW-2014. The authors thank Alessandro Sarti, 
Boris Khesin, and Yael Karshon for helpful discussions.}

\nocite{*}
\end{document}